\documentclass[preprint,12pt]{elsarticle}

\usepackage{amssymb}
\usepackage{amsmath}
\usepackage{amsthm}

 \usepackage{lineno}

\usepackage{physics}
\usepackage{subcaption}
\usepackage{siunitx}
\usepackage{url}
\usepackage{booktabs}
\usepackage{multirow}

\usepackage{tikz}
\usepackage{eso-pic}
\usepackage{xcolor}

\newtheorem{corollary}{Corollary}[section]

\newtheorem{proposition}{Proposition}[section]
\theoremstyle{definition}

\theoremstyle{remark}

\journal{Acta Astronautica}

\renewcommand\fbox{\fcolorbox{black}{gray!10}}

\newcommand\acceptedmanuscripttext{%
	\scriptsize
	\textbf{Accepted Manuscript.} 
	This manuscript has been accepted for publication in
	\textit{Acta Astronautica}.\\
	It has not undergone final copy-editing, typesetting, or proofreading and may be outdated.\\[5pt]
	The final Version of Record is available at
	\url{https://doi.org/10.1016/j.actaastro.2026.06.004} and is distributed under the terms of a CC BY-NC-ND 4.0 license.\\[5pt]
	Please cite the published version as:\\
	\textbf{A. Nunes, S. Brás, P. Batista, Nonlinear backstepping with saturation for low-thrust station-keeping of libration point orbits, Acta Astronautica 248 (2026) 324--342.}
}

\AddToShipoutPicture*{%
	\AtPageUpperLeft{%
		\begin{tikzpicture}[remember picture,overlay]
			\node[
			anchor=north,
			yshift=-10pt
			] at (current page.north)
			{%
				\fbox{%
					\parbox{\dimexpr\textwidth-2\fboxsep-2\fboxrule}{%
						\acceptedmanuscripttext
					}%
				}%
			};
		\end{tikzpicture}
	}%
}

\begin{document}

\begin{frontmatter}



\title{Nonlinear backstepping with saturation for low-thrust station-keeping of libration point orbits}


\author[ISR]{António Nunes\corref{ANcor}} 
\ead{antonio.w.nunes@tecnico.ulisboa.pt}
\cortext[ANcor]{Corresponding author at: Institute for Systems and Robotics, Instituto Superior Técnico, Universidade de Lisboa, Portugal}

\author[ESA]{Sérgio Brás} 
\author[ISR]{Pedro Batista} 

\affiliation[ISR]{organization={Institute for Systems and Robotics, Instituto Superior Técnico, Universidade de Lisboa},
	addressline={Av. Rovisco Pais 1}, 
	postcode={1049-001 },
	city={Lisbon},
	country={Portugal}}

\affiliation[ESA]{organization={European Space Agency},
	addressline={Keplerlaan 1, PO Box 299}, 
	postcode={2200AG},
	city={Noordwijk},
	country={The Netherlands}}

\begin{abstract}
This paper presents a novel nonlinear backstepping control law for continuous, low-thrust station-keeping in the Earth-Moon system. Quasi-periodic libration point orbits are targeted under a high-fidelity model of the dynamics. Almost global uniform exponential stability guarantees are attained, as shown through Lyapunov's stability theory. Saturation of the actuators is formally included in the controller design, such that these guarantees hold even in the event of saturation. The relationship between saturation threshold, control gains, and deviation is studied and an optimal procedure for gain selection is discussed. The control solution is tested numerically through a Monte Carlo analysis over representative application cases, subject to operational errors, constraints, and external perturbations. Station-keeping under actuation saturation is validated considering a conservative threshold for typical electric propulsion systems.
\end{abstract}




\end{frontmatter}



\section{Introduction}
\label{sec:intro}

 The interest on the operation of autonomous spacecraft in cislunar space has surged over recent years, following the renovated efforts of international space agencies, such as NASA and ESA, towards the re-consolidation of an active human presence on the Moon. In particular, the Artemis Program targets manned missions to the lunar surface as soon as 2027, following a 50 year hiatus since the pioneer endeavors of the Apollo program. These initiatives foment the development of novel spacecraft guidance, navigation, and control systems that are better prepared for the challenges ahead. 
 In regards to control, specifically, the desire to maximize mission span gives impetus to the design of increasingly efficient strategies that match the advancements in state-of-the-art propulsion systems.

Historically, most work has concentrated on the station-keeping of periodic orbits about the equilibrium points of the Circular Restricted Three-Body Problem (CR3BP) -- a well-known first approximation for the motion of spacecraft under the attraction of two massive celestial bodies, assumed to follow fixed circular orbits. Foundational contributions introduce concepts from feedback and linear optimal control \cite{farquhar1970control,breakwell1974Stationkeeping} and exploit natural dynamical symmetries through Floquet and invariant manifold theories, considering a linearization of the dynamics \cite{wiesel1983ModalControl,simo1986StationkeepingInvManif,simo1987OptimalControlHaloOrbits}. A particularly dominant line of research in classical literature is devoted to the study of impulsive control techniques that consider sudden changes in velocity. These significantly simplify the inclusion of optimization efforts, which is evidently relevant for the maximization of fuel efficiency and mission lifespan. For example, \textit{target-point} methods \cite{dwivedi1975,howell1993StationkeepingPointMethod} compute correction maneuvers through the minimization of a cost function weighting the magnitude of the actuation command with the projected satellite deviations over a number of discrete future instants, also under a linear viewpoint. Such techniques have also been complemented with global search, orbit continuation, and crossing control methods \cite{folta2010StationKeeping,pavlak2012Strategy,folta2013StationKeeping}, aiming to leverage optimal maneuvers with trajectory design by targeting desired conditions several revolutions downstream.

Despite the research interest of the aforementioned approaches, the insufficient degree of realism provided by the CR3BP for most practical applications means that the use of higher-fidelity models of cislunar dynamics is required. More specifically, it is typically necessary for these models to depict the motion of the celestial bodies through real ephemeris data, earning them the classification of high-fidelity ephemeris models (HFEMs). In addition, these models usually incorporate the gravitational contributions of other celestial bodies besides the two primaries, as well as relevant perturbations acting on the spacecraft. Station-keeping under HFEMs has also been tackled, for the most part, through impulsive control techniques that stem from those applied to the CR3BP or its derivatives. For this reason, it is common to find target-point \cite{zhang2022StationKeepingHFEM}, multiple-shooting optimization \cite{muralidharan2020ControlSF,ghorbani2013ContinuousAndImpulsive}, and crossing control \cite{petersen2019JWST} strategies in literature. Alternatively, other authors have assessed impulsive approaches anchored on discrete sliding mode control \cite{lian2014SlidingMode}, discrete linear quadratic regulation (LQR) \cite{zhang2022StationKeepingHFEM,lian2014SlidingMode}, Cauchy-Green tensor targeting \cite{Guzzetti2017stationkeeping}, and the exploitation of stretching directions \cite{muralidharan2022StretchingDirections}.

However, the demanding control commands requested by discrete solutions at each impulse make it imperative to utilize high-trust propulsion systems, typically based on chemical or cold-gas thrusters. Ultimately, such strategies preclude the use of novel electrical propulsion technologies that have shown to be superior in terms of efficiency, albeit limited to commands of relatively low magnitude  \cite{shastry2025ElectricPropulsion, Herman2016IonPropulsion}. To address this concern, recent research has focused on assessing the use of continuous strategies for the station-keeping tasks of future missions.
These solutions are also varied and have been applied to both the CR3BP and HFEMs. To this end, some authors have focused on the use of well-established control approaches, such as LQR \cite{qi2019ContinuousThrust,qi2022ContinuousThrust,nazari2014LQR_Backstepping,ghorbani2013ContinuousAndImpulsive}, introducing optimality from a linearization standpoint to reduce computational demands. Robust alternatives consider the use of  $H_2$ and $H_\infty$ control \cite{jones1993H2,kulkrani2006}, to maximize disturbance rejection. Nonlinear control has also been implemented, namely through output regulation for formation flying \cite{akiyama2018outputRegulation}, backstepping to achieve formal guarantees for stability \cite{Nunes2025backstepping}, and nonlinear programming for improved optimality, however requiring the calculation of the system's Hamiltonian \cite{ghorbani2013ContinuousAndImpulsive}. Linear programming (LP) has been studied as an extension of classic, impulsive multiple-shooting approaches, through the construction of large-scale optimization problems \cite{ulybyshev2015LinearProgramming}.
More recently, other works have sought modern alternatives, such as model predictive control (MPC) \cite{Elobaid2022MPC,aerospace9120798}, to incorporate also physical constraints in the design of the control response, albeit at an increase in computational demands. A survey of continuous (and impulsive) station-keeping strategies, including but not limited to the ones detailed above, is found in \cite{shirobokov2017Survey}.

Despite the wide variety of continuous station-keeping techniques in current literature, achieving a cohesive solution that harmonizes numerical simplicity, formal stability guarantees, and the inclusion of physical constraints remains an open challenge. The goal of this paper is to address this challenge and to develop a computationally efficient strategy that bridges theory with practice, balancing typical limitations of electric propulsion systems with rigorous guarantees of convergence. To this end, this paper proposes a novel control solution for orbital station-keeping, based on continuous nonlinear backstepping, that formally incorporates an actuation saturation constraint at the design stage. To do so, a base formulation without saturation is initially presented and almost global uniform exponential stability guarantees under HFEM dynamics are formalized. Then, by modifying the base control law, structural conditions on the control gains are derived through Lyapunov's theory on stability, showing that it is possible to further respect the saturation constraint without a sacrifice to the stability guarantees attained. This strategy goes a step beyond current literature as: \textit{(i)} saturation is formally included in the controller design phase and does not put into question stability guarantees under the full, nonlinear HFEM; and, simultaneously, \textit{(ii)} computational demands are kept at bay, as optimal gains may be calculated offline and require no on-board optimization efforts. For these reasons, the proposed approach provides meaningful advancements to the state-of-the-art in low-thrust station-keeping -- where the inclusion of actuation saturation is typically seen as an \textit{a posteriori} concern, destroying stability guarantees, or requires computationally expensive real-time optimization to be met, e.g. via LP or MPC strategies. Given their current relevance, we focus the application of the proposed solution to quasi-periodic orbits (QPOs) of the Earth-Moon system, whose dynamics are modeled through an HFEM considering meaningful external perturbations acting on the spacecraft. For improved realism, further operational constraints and errors are introduced at the navigation and control levels. The proposed strategy is evaluated through the means of numerical simulations over test cases of relevance and compared to other continuous and impulsive strategies from the literature. Moreover, the practical effects of formally including actuation saturation are evaluated, and a procedure for optimal gain and insertion point selection is discussed.


The remainder of this work is structured as follows. In Section~\ref{sec:dyn}, the CR3BP and HFEM dynamics of the Earth-Moon system are briefly covered, drawing attention to key differences in the types of trajectories that may be found. Section~\ref{sec:con} provides the formal derivation of the controller, starting with the base nonlinear backstepping law and its extension considering the inclusion of actuation saturation. Stability guarantees are assessed and formalized making recourse of Lyapunov's theory. The control law is then evaluated numerically in Section~\ref{sec:res} in three steps. First, the base control law is compared with continuous and impulsive alternatives from literature, considering objective benchmarks subject to a Monte Carlo analysis. Then, its adequacy at rejecting relevant external perturbations is assessed and station-keeping performance is stacked against the unperturbed case. Lastly, the meanders of formally including actuation saturation are discussed in practical terms, offering key insights on the gain decision process and insertion point selection through numerical means. Control under saturation is further validated through simulations. Finally, in Section~\ref{sec:conc}, the most meaningful takeaways from this research and the proposed method are summarized, and possible paths for future investigation are discussed.

\section{Dynamics}
\label{sec:dyn}

As previously detailed, this work deals exclusively with station-keeping under realistic HFEM dynamics, aiming to meet stringent demands of ongoing and future missions within cislunar space. Nonetheless, it is meaningful to briefly cover also the CR3BP formulation for its dynamical insights and characteristic scales. These are later used to scale the HFEM model and simplify the representation of the target trajectories in space. Moreover, it is relevant to discuss the differences in the types of solutions that exist in both dynamical models and how they translate from one to another -- a key problem that has been discussed extensively in the literature \cite{Nunes2026TrajectoryDesign}.

\subsection{CR3BP}
\label{subsec:cr3bp}

Restricted three-body problems (R3BPs) are a subclass of three-body dynamics that consider one of the masses to be negligible in comparison with the other two. In that case, the motion of the two massive bodies, named primaries, may be interpreted as a solution of an isolated two-body problem (2BP) \cite{szebehely1967TheoryOfOrbits}. Of particular relevance is the case where the primaries are assumed to describe circular orbits about their shared barycenter, which gives rise to the CR3BP. This provides a model for spacecraft motion near two celestial bodies that is amply considered in literature for its mathematical simplicity but dynamical complexity and richness. Moreover, it constitutes a reasonable first approximation for most two-body configurations of interest in the Solar System, namely the Earth-Moon, Sun-Earth, or Sun-Jupiter binaries.

We briefly introduce the typical CR3BP formulation in order to cover relevant topics that will be of use for the more realistic dynamical model considered in this work -- interested readers are directed to \cite{szebehely1967TheoryOfOrbits,ross3BPbook} for a more thorough exposition of the CR3BP dynamics. As the bodies are assumed to describe circular orbits about the shared barycenter, it is common to consider a \textit{synodic} reference frame that aligns both primaries along the $X$ axis and rotates along $Z$ to match the angular speed along their respective orbits. Finally, $Y$ completes the right-handed reference frame. Moreover, the distance between primaries is usually scaled to unity and their orbital period to $2\pi$, resulting in unit angular velocity. Denoting, without loss of generality, the most massive primary with subscript $1$ and the other by subscript $2$, their masses are usually scaled according to a problem-specific constant $\mu$ such that ${m_1=1-\mu}$ and $m_2=\mu$.  For the Earth-Moon system, considered in this work, $\mu=0.01215$. As a result of the aforementioned procedure, the location of the primaries, modeled as point-masses, is fixed in the synodic frame to $\mathbf{r}_1=(-\mu,0,0)$ and $\mathbf{r}_2=(1-\mu,0,0)$. In that case, the effective potential that governs the motion of the third point-mass (the spacecraft), whose position we denote by $\mathbf{r}^T(t)=\begin{bmatrix}
	x(t) & y(t) & z(t)
\end{bmatrix}$, at each instant $t$, may be established in the synodic reference frame as
\begin{equation}
	\label{eq:eff_potential}
	U = \frac{x^2(t)+y^2(t)}{2} + \frac{1-\mu}{\norm{\mathbf{r}-\mathbf{r}_1}} + \frac{\mu}{\norm{\mathbf{r}-\mathbf{r}_2}},
\end{equation}
being comprised of a Coriolis acceleration and two gravitational contributions \cite{ross3BPbook}. Note that the positions of the primaries constitute singularities in the dynamics, where $U$ is not defined.

Adopting the dot notation to represent derivatives with respect to time, the equations of motion (EoM) of the spacecraft in the synodic reference frame, over the domain of definition of $U$, become
\begin{equation}
	\label{eq:CR3BP_dynamics}
	\begin{aligned}
		\dot{\mathbf{r}}(t) &= \mathbf{I}_3\mathbf{v}(t), \\
		\dot{\mathbf{v}}(t) &= \boldsymbol{\Omega} \mathbf{v}(t) + \grad{U(\mathbf{r}(t))},
	\end{aligned}
	\qquad \text{with} \qquad
	\boldsymbol{\Omega} = \begin{bmatrix}
		0 & 2 & 0\\
		-2 & 0 & 0\\
		0 & 0 & 0
	\end{bmatrix},
\end{equation}
where $\mathbf{v}^T(t) = \begin{bmatrix}
	\dot x(t) & \dot y(t) & \dot z(t)
\end{bmatrix}$ denotes the spacecraft's velocity, $\mathbf{I}_3$ is the $3\times 3$ identity matrix, and $\grad{U(\mathbf{r}(t))}$ is the gradient of the effective potential in Eq.~\eqref{eq:eff_potential}, evaluated at $\mathbf{r}(t)$, \cite{ross3BPbook}.

The CR3BP dynamics in Eq.~\eqref{eq:CR3BP_dynamics} constitute a prominent example of a chaotic system of differential equations in dynamical systems theory, with small differences in initial conditions leading to vastly different (and most times unpredictable) solutions over time. To this end, a general solution for the CR3BP cannot be found \cite{szebehely1967TheoryOfOrbits}. Still, it is possible to show that these EoM have five equilibrium points, so-called Lagrange or libration points, which are ordered $L_1$ through $L_5$ by effective potential.  All of these points lie on the orbital plane of the primaries: the first three along the $X$ axis and the remaining two forming equilateral triangles with the primaries on both sides of said axis. Moreover, given that the dynamics are \textit{autonomous} (i.e., they do not depend explicitly on time), another classical result is that they admit an infinite number of periodic solutions in phase space (position and velocity), as first approached by Poincaré \cite{poincare1892MethodesNouvelles}. A subset of these solutions that is of particular relevance is that of periodic orbits that revolve predominantly about one of the Lagrange points, considered optimal candidates for many missions within the Earth-Moon and Earth-Sun systems, for example, due to their placement with respect to the celestial bodies and communications clearance. Recent examples of spacecraft taking advantage of such orbits in the aforementioned systems include the James-Webb Space Telescope (JWST) \cite{gardner2006JWST}, the Interstellar Mapping and Acceleration Probe (IMAP) \cite{McComas2025IMAP}, and the Chang'e-4 spacecraft mission \cite{li2021ChangE}.

\subsection{HFEM}
\label{subsec:hfem}

When it comes to the HFEM, it is common to consider the influence of $P>2$ celestial bodies -- i.e., beyond the two of the CR3BP. To this end, the dynamics of the spacecraft are typically established in an inertial frame centered at one of the celestial bodies considered. In this work, we select, without loss of generality, $j=1$ to be the Earth, such that the dynamics may be established in the well-known J2000 inertial reference frame. Under the standard assumption that the celestial bodies may be adequately modeled as point-masses and that no external perturbations act on the spacecraft, the resulting (non-actuated) spacecraft EoM may be written as
\begin{equation}
	\label{eq:HFEM_dynamics}
	\begin{aligned}
		\dot{\mathbf{r}}(t) &= \mathbf{I}_3 \mathbf{v}(t), \\
		\dot{\mathbf{v}}(t) &=  -\mu_1\frac{\mathbf{r}(t)}{\norm{\mathbf{r}(t)}^3} + \sum_{j=2}^{P} \mu_j\left(\frac{\mathbf{r}_j(t)-\mathbf{r}(t)}{\norm{\mathbf{r}_j(t)-\mathbf{r}(t)}^3} - \frac{\mathbf{r}_j(t)}{\norm{\mathbf{r}_j(t)}^3}\right),
	\end{aligned}
\end{equation}
where $\mu_j$ denotes gravitational parameter of the $j$-th celestial body and $\mathbf{r}_j$ is its position in the J2000 inertial frame, such that $\mathbf{r}_1(t)=\mathbf{0},~\forall t$.
Similarly to the CR3BP, these dynamics are defined over $\mathbf{r}(t)\in\mathbb{R}^3\setminus\mathcal{S}(t)$ and $\mathbf{v}(t)\in\mathbb{R}^3$, where
\begin{equation}
	\label{eq:set_singularities}
	\mathcal{S}(t):=\left\{ \mathbf{x}\in\mathbb{R}^3: \mathbf{x}=\mathbf{r}_j(t),~  j=1,\dots,P\right\}
\end{equation}
represents the set of position singularities at time $t$, i.e. the location of the celestial bodies considered. For the sake of consistency, we consider the dynamics in Eq.~\eqref{eq:HFEM_dynamics} to be scaled according to the characteristic length, time, and mass scales of the CR3BP, from Section~\ref{subsec:cr3bp}.

We note that, due to the explicit time-dependence of the celestial bodies' position, $\mathbf{r}_j(t)$, no change of coordinates can make the HFEM EoM autonomous for $P>2$. Hence, when compared to the CR3BP, this dynamical model has no constants of motion or relevant symmetries, inhibiting the existence of truly periodic solutions. However, it is also the case that the HFEM is often interpreted as a perturbation of the CR3BP \cite{gomez2002SolarSystemFrequencies}, such that quasi-periodic trajectories may be found in the vicinity of CR3BP periodic orbits, provided that relevant Hamiltonian structures are preserved \cite{AlmanzaSoto2025PersistenceOfStructures}. Exceptions may be found in transition-challenging regions of phase space or for periodic orbits close to resonance with the period of the dominant bodies \cite{park2025CharacterizationOfL2Analogs,sanaga2024ChallengingRegion} but, for most general purposes, a trajectory equivalent to a periodic orbit from the CR3BP may typically be found in the HFEM. We take advantage of this fact to transition periodic solutions from the CR3BP into QPOs in the HFEM, which serve as target trajectories for station-keeping purposes. To do so, a method developed in previous work is utilized \cite{Nunes2026TrajectoryDesign}.

Given a target trajectory, it is useful to write the dynamics in error form with the aim of developing a control strategy that ensures accurate tracking. In this sense, consider a nominal point in phase space that evolves along the target solution with time, denoted $(\mathbf{r}^*(t),\mathbf{v}^*(t))$, and let $\mathbf{z}_1(t) := \mathbf{r}(t)-\mathbf{r}^*(t)$ and $\mathbf{z}_2(t) := \mathbf{v}(t)-\mathbf{v}^*(t)$ be the position and velocity deviations from this point, respectively. Then, we may write the controlled error dynamics as
\begin{equation}
	\label{eq:HFEM_variational}
	\begin{aligned}
		\dot{\mathbf{z}}_1(t) &= \mathbf{z}_2(t),\\
		\dot{\mathbf{z}}_2(t) & = \mathbf{f}_a(t,\mathbf{z}_1(t)) + \mathbf{u}(t),
	\end{aligned}
\end{equation}
where $\mathbf{u}(t)$ is a control acceleration and
\begin{equation*}
	\begin{aligned}
		\mathbf{f}_a(t,\mathbf{z}_1(t)) :=& \sum_{j=1}^{P}\mu_j \left(\frac{\mathbf{r}_j(t)-\mathbf{r}(t)}{\norm{\mathbf{r}_j(t)-\mathbf{r}(t)}^3}- \frac{\mathbf{r}_j(t)-\mathbf{r}^*(t)}{\norm{\mathbf{r}_j(t)-\mathbf{r}^*(t)}^3}\right)\\
		=&\sum_{j=1}^{P}\mu_j \left(\frac{\mathbf{w}_j(t)-\mathbf{z}_1(t)}{\norm{\mathbf{w}_j(t)-\mathbf{z}_1(t)}^3} - \frac{\mathbf{w}_j(t)}{\norm{\mathbf{w}_j(t)}^3}\right)
	\end{aligned}
\end{equation*}
is the acceleration error induced by the deviation from the target trajectory, with $\mathbf{w}_j(t) := \mathbf{r}_j(t)-\mathbf{r}^*(t)$.
Naturally, the set of singularities previously identified in Eq.~\eqref{eq:set_singularities} is maintained. Conversely, we may define the domain of definition of the vector field of Eq.~\eqref{eq:HFEM_variational} by $(\mathbf{z}_1(t), \mathbf{z}_2(t))\in\mathcal{D}(t)$, where
\begin{equation*}
	 \mathcal{D}(t):=\left\{(\mathbf{x},\mathbf{y})\in\mathbb{R}^6: \mathbf{x} \neq \mathbf{w}_j(t),~j=1,\dots,P \right\}
\end{equation*}
corresponds to the set of all points in phase space of Eq.~\eqref{eq:HFEM_dynamics} whose position $\mathbf{r}(t)$ does not coincide with one of the celestial bodies at time $t$, $\mathbf{r}_j(t)$. This notation proves useful when establishing the propositions presented in the following section.

\section{Controller Formulation}
\label{sec:con}

In this work, the orbital station-keeping problem under HFEM dynamics is tackled through a strategy derived from nonlinear control theory and Lyapunov's concepts on stability\footnote{A commendable introduction to these topics is found in \cite{khalil2001NLSystems}.}. In particular, we resort to the well-known backstepping technique to develop a base control law that ensures uniform exponential stability (UES), in the sense of Lyapunov, of the origin of the dynamics in Eq.~\eqref{eq:HFEM_variational}. By \textit{uniform}, we mean that these guarantees hold for all $t\geq t_0$, where $t_0\geq0$ is a given initial instant. Without loss of generality, we assume that time is established relatively to a chosen epoch, such that $t_0=0$. We show that the guarantees provided by the nonlinear control law are \textit{almost global}, in the sense that they hold for all initial conditions and trajectories, apart from a set of measure zero. Later, the domain of analysis is restricted to allow for the formal inclusion of an actuation saturation constraint, without a sacrifice to the exponential nature of the stability guarantees attained over said domain.

\subsection{Base Control Law}
\label{sec:base_law}

In regards to the base control law, we extend the formulation devised in previous work for low-fidelity models \cite{Nunes2025backstepping} to the HFEM dynamics, clarifying the nature of the stability guarantees attained. This is formalized in the following statement.
\begin{proposition}
	\label{prop:base_backstepping}
	Consider the HFEM dynamics in Eq.~\eqref{eq:HFEM_variational}. Let $\mathbf{r}^*(t)\in\mathbb{R}^3\setminus\mathcal{S}(t)$ and $\mathbf{v}^*(t)\in \mathbb{R}^3$ denote the position and velocity states of a target trajectory, over time $t\geq0$. Recall that $\mathbf{z}_1(t) =\mathbf{r}(t)-\mathbf{r}^*(t)$ and ${\mathbf{z}_2(t)=\mathbf{v}(t)-\mathbf{v}^*(t)}$ denote the spacecraft position and velocity deviations, respectively, at each instant. Then, the control law
	\begin{equation}
		\label{eq:base_law}
		\mathbf{u}(t) = - (\mathbf{I}_3 + \mathbf{K}_2 \mathbf{K}_1) \mathbf{z}_1(t) - (\mathbf{K}_1 + \mathbf{K}_2) \mathbf{z}_2(t) - \mathbf{f}_a(t,\mathbf{z}_1(t)),
	\end{equation}
	where $\mathbf{K}_1$ and $\mathbf{K}_2$ are constant (symmetric) positive definite gain matrices, guarantees almost global uniform exponential stability of the origin of Eq.~\eqref{eq:HFEM_variational}.
\end{proposition}
\begin{proof}
	We divide the proof into two parts. Firstly, we show that the origin of the controlled system in Eq.~\eqref{eq:HFEM_variational} is uniformly exponentially stable over the domain of definition, i.e. $\mathbf{z}(t) \in \mathcal{D}(t)\subset \mathbb{R}^6$. Then, we show that exponential convergence towards the origin holds for almost all trajectories beginning in $\mathbb{R}^6$. In other words, we show that the set of initial conditions from which trajectories start outside or eventually exit the domain of definition of the vector field have measure zero in $\mathbb{R}^6$.
	
	The proof for the first statement follows from meeting the conditions of \cite[Theorem 4.10]{khalil2001NLSystems}. Namely, it requires the existence of a Lyapunov function $V(t,\mathbf{z}(t))$ that satisfies, $\forall t\geq 0$ and $\mathbf{z}^T(t)=\begin{bmatrix}
		\mathbf{z}_1^T(t) & \mathbf{z}_2^T(t)
	\end{bmatrix}\in\mathcal{D}(t)$, 
	\begin{equation*}
			a \norm{\mathbf{z}(t)}^\gamma \leq V(t,\mathbf{z}(t)) \leq b \norm{\mathbf{z}}^\gamma \quad \text{and} \quad \dot{V}(t,\mathbf{z}(t)) \leq - c \norm{\mathbf{z}(t)}^\gamma,
	\end{equation*}
	where $a$, $b$, $c$, and $\gamma$ are positive constants.
	
	As such, consider the time-invariant candidate Lyapunov function
	\begin{equation}
		\label{eq:base_lyap}
		\begin{aligned}
			V(\mathbf{z}(t)) &= \frac{1}{2}\mathbf{z}_1^T(t) \mathbf{z}_1(t) +  \frac{1}{2}\left[\mathbf{z}_2(t) + \mathbf{K}_1\mathbf{z}_1(t)\right]^T\left[\mathbf{z}_2(t) + \mathbf{K}_1\mathbf{z}_1(t)\right] \\ &= \frac{1}{2}\mathbf{z}^T(t)\mathbf{X}\mathbf{z}(t),
		\end{aligned}
	\end{equation}
	where
	\begin{equation*}
		\mathbf{X} = \begin{bmatrix}
			\mathbf{I}_3 + \mathbf{K}_1 \mathbf{K}_1 & \mathbf{K}_1\\
			\mathbf{K}_1 & \mathbf{I}_3
		\end{bmatrix} =\mathbf{X}^T.
	\end{equation*}
	For $\mathbf{K}_1\succ0$, Eq.~\eqref{eq:base_lyap} clearly satisfies $V(\mathbf{z}(t))>0,~\forall \mathbf{z}(t)\neq \mathbf{0}$, such that 
	\begin{equation}
		\label{eq:X_aux}
		\frac{1}{2}\lambda_\text{min} (\mathbf{X}) \norm{\mathbf{z}(t)}^2 \leq V(\mathbf{z}(t)) \leq \frac{1}{2}\lambda_\text{max}(\mathbf{X})\norm{\mathbf{z}(t)}^2, ~\forall \mathbf{z}(t)\in \mathbb{R}^6\supset \mathcal{D}(t).
	\end{equation}
	For $\mathbf{z}(t)\in\mathcal{D}(t)$, the time-derivative of the candidate Lyapunov function from Eq.~\eqref{eq:base_lyap} may be written as
	\begin{equation*}
			\dot{V}(\mathbf{z}(t)) = \mathbf{z}_1^T(t)\mathbf{z}_2(t) + \left[\mathbf{z}_2(t) + \mathbf{K}_1\mathbf{z}_1(t)\right]^T\left[\mathbf{K}_1\mathbf{z}_2(t)+\mathbf{f}_a(t,\mathbf{z}_1(t))+\mathbf{u}(t)\right],
	\end{equation*}
	after substitution of the dynamics from Eq.~\eqref{eq:HFEM_variational}. The application of the control law leads to the quadratic form
	\begin{equation}
		\label{eq:base_lyap_dot}
		\begin{aligned}
			\dot{V}(\mathbf{z}(t)) &= -\mathbf{z}_1(t)^T \mathbf{K}_1 \mathbf{z}_1(t) - \left[\mathbf{z}_2(t) + \mathbf{K}_1\mathbf{z}_1(t)\right]^T\mathbf{K}_2\left[\mathbf{z}_2(t) + \mathbf{K}_1\mathbf{z}_1(t)\right] \\ &= -\mathbf{z}^T(t)\mathbf{Y}\mathbf{z}(t),
		\end{aligned}
	\end{equation}
	where
	\begin{equation*}
		\mathbf{Y} = \begin{bmatrix}
			\mathbf{K}_1 + \mathbf{K}_1\mathbf{K}_2\mathbf{K}_1 & \mathbf{K}_1\mathbf{K}_2\\
			\mathbf{K}_2\mathbf{K}_1 & \mathbf{K}_2
		\end{bmatrix}=\mathbf{Y}^T.
	\end{equation*}
	Once more, for $\mathbf{K}_1,\mathbf{K}_2\succ0$, Eq.~\eqref{eq:base_lyap_dot} satisfies $\dot{V}(\mathbf{z}(t))<0,~\forall \mathbf{z}(t)\neq \mathbf{0}$, such that
	\begin{equation}
		\label{eq:dot_V_aux}
		\dot{V}(\mathbf{z}(t)) \leq - \lambda_\text{min}(\mathbf{Y})\norm{\mathbf{z}(t)}^2, ~\forall\mathbf{z}(t)\in\mathcal{D}(t).
	\end{equation}
	To this end, the Lyapunov function in Eq.~\eqref{eq:base_lyap} verifies the conditions of \cite[Theorem 4.10]{khalil2001NLSystems} with the positive constants
	\begin{equation*}
		a=\frac{1}{2}\lambda_\text{min}(\mathbf{X}), \quad b=\frac{1}{2}\lambda_\text{max}(\mathbf{X}), \quad c=\lambda_\text{min}(\mathbf{Y}), \quad  \text{and} \quad \gamma=2,
	\end{equation*} 
	such that the origin of the dynamics in Eq.~\eqref{eq:HFEM_variational} is a uniformly exponentially stable equilibrium point, by application of the control law in Eq.~\eqref{eq:base_law}.
	
	To prove the second statement, we first note that the conditions at $t=0$ uniquely define a solution over time, assuming it starts and remains within the domain of definition of the dynamics in Eq.~\eqref{eq:HFEM_variational}. Hence, to show that the UES guarantees are almost global, it is sufficient to show that the set of initial conditions where these guarantees fail has measure zero in $\mathbb{R}^6$. Given the exponential stability of the flow and the linear nature of the controlled dynamics in Eq.~\eqref{eq:HFEM_variational}, we note that the only exceptions are therefore: \textit{(i)} initial conditions that coincide exactly with the singularities' position at $t=0$; or \textit{(ii)} initial conditions leading to an intersection with any singularity in finite time. Clearly, since $\mathcal{D}(0)$ is simply $\mathbb{R}^6$ devoid of a finite amount of point singularities in position, the set of initial conditions that coincide with the singularities at $t=0$ has measure zero in $\mathbb{R}^6$. In regards to the initial conditions that lead to a collision, time-dependence makes the analysis non-trivial. Nonetheless, in \ref{ap:collision_traj_proof} we show that these also constitute a set of measure zero in $\mathbb{R}^6$, proving that the UES guarantees attained are almost global in nature.
\end{proof}

By restricting the domain of analysis, it is possible to completely avoid the singularities in the dynamics and attain full guarantees of exponential stability over the entirety of said domain, as per the following statement.
\begin{corollary}[of Proposition~\ref{prop:base_backstepping}]
	\label{cor:ell_restriction}
	Consider the conditions of Proposition~\ref{prop:base_backstepping} and define the open ellipsoid
	\begin{equation}
		\label{eq:ell_set}
		\mathcal{E} = \left\{\mathbf{y} \in \mathbb{R}^6: \mathbf{y}^T\mathbf{X}\mathbf{y} < \lambda_{\min}(\mathbf{X}) r^2\right\}, \quad \text{with} \quad r = \min_{\substack{j=1,\dots,P\\ t\in[0, t_s]}} \norm{\mathbf{w}_j(t)},
	\end{equation}
	where $t_s>0$ is the final instant of analysis. Then, all trajectories beginning at $\mathbf{z}_0:=\mathbf{z}(0) \in \mathcal{E}$ converge exponentially fast to the origin of Eq.~\eqref{eq:HFEM_variational}, over $t\in[0,t_s]$.
\end{corollary}
\begin{proof}
	Since it has been established that $\mathbf{X}\succ0$, one has, by construction,
	\begin{equation*}
		\mathbf{y}^T\mathbf{X}\mathbf{y}\geq \lambda_{\min}(\mathbf{X})\norm{\mathbf{y}}^2,~\forall \mathbf{y}\in \mathbb{R}^6,
	\end{equation*}
	such that
	\begin{equation*}
		\norm{\mathbf{y}}^2\leq \frac{\mathbf{y}^T \mathbf{X} \mathbf{y} }{\lambda_{\min}(\mathbf{X})}< r^2,~\forall \mathbf{y}\in \mathcal{E}.
	\end{equation*}
	The ellipsoid $\mathcal{E}$ is thus a subset of the largest open $\ell_2$ norm-ball, centered at the origin, that excludes the closest singularity over $t\in[0,t_s]$. It is therefore possible to state that $\mathcal{E}\subset \mathcal{D}(t),~\forall t \in [0, t_s]$. Moreover, $\mathcal{E}$ constitutes a strict sub-level set of the Lyapunov function from Eq.~\eqref{eq:base_lyap}, i.e.
	\begin{equation*}
		\mathcal{E} = \{\mathbf{y}\in \mathbb{R}^6: V(\mathbf{y})< \frac{1}{2}\lambda_{\min}(\mathbf{X})r^2\}.
	\end{equation*}
	
	Satisfying the conditions of Proposition~\ref{prop:base_backstepping} ensures the time-derivative of the Lyapunov function is negative definite in $\mathcal{D}(t),~\forall t\geq 0$, such that $V(\mathbf{z}(t))$ is strictly decreasing for $\mathbf{z}(t)\in\mathcal{E}\subset \mathcal{D}(t)$. Trajectories starting inside $\mathcal{E}$ must thus remain in this sub-level set over ${t\in[0,t_s]}$, i.e. $\mathcal{E}$ is (strictly) \textit{forward-invariant}. Since $\mathcal{E}$ excludes all singularities of the dynamics in Eq.~\eqref{eq:HFEM_variational} over ${t\in[0,t_s]}$ and contains the equilibrium point at the origin, trajectories beginning inside this set are guaranteed to converge exponentially fast to $\mathbf{z}=\mathbf{0}$ during the time interval under analysis, following Proposition~\ref{prop:base_backstepping}.
\end{proof}

We highlight that the constraint in time, $t\in[0,t_s]$, in Corollary~\ref{cor:ell_restriction} is mostly a formality due to the time-dependence and limited availability of planetary ephemeris data. Given that missions usually have a limited duration, this concern is inconsequential in practical terms. Still, for nearly-repeating trajectories such as QPOs, it may be the case that $\min_{j=1,\dots,P}\norm{\mathbf{w}_j(t)}$ may be safely bounded for all $t\geq0$, making it possible to formalize UES guarantees.

\subsection{Control Law with Actuation Saturation}
\label{subsec:con_sat}

In order to improve the applicability of the proposed station-keeping strategy in realistic scenarios, we ponder the introduction of physical constraints to the operation of the actuators. In particular, given that this strategy has been developed with the intent of enabling the use of low-thrust electrical propulsion systems, a logical constraint to be included is that of an actuation saturation threshold. In this paper, we proceed under the assumption that the control command can be oriented freely in $3$-dimensional space and that the saturation constraint acts on the entire control command. In other words, we consider that its magnitude is upper bounded by $\norm{\mathbf{u}(t)}\leq u_\text{sat}$, where $u_\text{sat}>0$ is the established saturation threshold -- a physical characteristic of the specific propulsion technology employed. This assumption is on par with other continuous strategies in the literature which, however, do not typically incorporate the saturation constraint formally in the controller design, e.g. \cite{qi2019ContinuousThrust}.

To meet the saturation constraint, we consider that the original control law is relieved at the level of the linear term, i.e.
\begin{equation}
	\label{eq:relieved_law}
	\mathbf{u}(t) = -\beta(t) \mathbf{K} \mathbf{z}(t) - \mathbf{f}_a(t,\mathbf{z}_1(t)), ~ \text{with} ~ \mathbf{K} := \begin{bmatrix}
		\mathbf{I}_3 + \mathbf{K}_2\mathbf{K}_1 &
		\mathbf{K}_1 + \mathbf{K}_2
	\end{bmatrix},
\end{equation}
where $\beta(t) \in [0,1]$ is chosen such that $\norm{\mathbf{u}(t)}\leq u_\text{sat},~\forall t\geq 0$. Evidently, this approach immediately restricts the analysis to saturation thresholds larger than the maximum foreseeable value of $\norm{\mathbf{f}_a(t,\mathbf{z}_1(t))}$, which is acknowledged as a limitation that grants further discussion at a later point. For this reason, the quality of the upper bound for the nonlinear acceleration error plays an important part in this study, as later discussed.

Given that, by application of the control law from Eq.~\eqref{eq:relieved_law}, the nonlinear term of the EoM in Eq.~\eqref{eq:HFEM_variational} is eliminated, the resulting dynamics are linear with respect to $\mathbf{z}(t)$, i.e.
\begin{equation}
	\label{eq:LTV_dynamics}
	\dot{\mathbf{z}}(t) = \begin{bmatrix}
		\mathbf{0}_3 & \mathbf{I}_3 \\
		-(\mathbf{I}_3 + \mathbf{K}_2 \mathbf{K}_1)\beta(t) & -(\mathbf{K_1} + \mathbf{K}_2)\beta(t)
	\end{bmatrix} \mathbf{z}(t).
\end{equation}
Moreover, if the gain matrices are selected as $\mathbf{K}_1 = k_1 \mathbf{I}_3$ and $\mathbf{K}_2 = k_2 \mathbf{I}_3$, it is possible to separate the dynamics coordinate-wise into three $2$-dimensional systems in the form of
\begin{equation}
	\label{eq:LTV_dynamics_separate}
	\begin{bmatrix}
		\dot{\mathbf{z}}_1^{(i)}(t)\\
		\dot{\mathbf{z}}_2^{(i)}(t)
	\end{bmatrix}= \begin{bmatrix}
		0 & 1\\
		-(1+k_1 k_2)\beta(t) & -(k_1 + k_2)\beta(t)
	\end{bmatrix}\begin{bmatrix}
		\mathbf{z}_1^{(i)}(t)\\
		\mathbf{z}_2^{(i)}(t)
	\end{bmatrix},
\end{equation}
where $\mathbf{y}^{(i)}$ denotes the $i$-th entry of $\mathbf{y}$. In this case, verifying stability of the original system from Eq.~\eqref{eq:LTV_dynamics} is equivalent to verifying it for the individual systems along each coordinate. Since $\beta(t)$ is expected to vary with time to ensure respect for the saturation constraint, the dynamics in Eq.~\eqref{eq:LTV_dynamics_separate} constitute a linear time-varying system (LTV). Hence, in spite of its linear nature, the analysis of the system's poles is insufficient to infer on the stability of Eqs.~\eqref{eq:LTV_dynamics_separate} and \eqref{eq:LTV_dynamics}. 

By substitution of the relieved control law from Eq.~\eqref{eq:relieved_law} in the Lyapunov candidate function in Eq.~\eqref{eq:base_lyap}, assuming $\mathbf{K}_1=k_1\mathbf{I}_3$ and $\mathbf{K}_2=k_2\mathbf{I}_3$, one gets, after some simplifications,
\begin{equation*}
	\dot{V}(\mathbf{z}(t)) = -\frac{1}{2}\mathbf{z}^T(t) \mathbf{U} \mathbf{z}(t),
\end{equation*}
where
\begin{equation*}
	\mathbf{U} = \begin{bmatrix}
		2( k_1 +  k_1^2 k_2)\beta(t) \mathbf{I}_3 & \left[(k_1^2+ 2k_1 k_2 +1)\beta(t) -k_1^2 -1\right]\mathbf{I}_3\\
		\left[ (k_1^2+ 2k_1 k_2 +1)\beta(t) -k_1^2 -1\right]\mathbf{I}_3 & 2\left[(k_1 +  k_2)\beta(t) - k_1\right] \mathbf{I}_3
	\end{bmatrix}
\end{equation*}
is symmetric.
For the purposes of ensuring stability, we are interested in evaluating the limit conditions on $\beta(t)$ that guarantee $\mathbf{U}$ is positive definite. In particular, we may exploit the separability of the original system into its independent $2\times2$ constituents. Since each constituent is symmetric, a sufficient condition for positive definiteness is given by Sylvester's criterion,
\begin{equation*}
	\left\{\begin{aligned}
		& (k_1  + k_1^2 k_2)\beta(t) > 0, \\
		&4 \left[(k_1 + k_1^2 k_2)\beta(t)\right]\left[(k_1+k_2)\beta(t)-k_1\right] - \left[(k_1^2+ 2k_1 k_2 +1)\beta(t) -k_1^2 -1\right]^2>0 .
	\end{aligned}\right.
\end{equation*}
The first condition simply restricts $\beta(t)$ to be a positive value. The second condition further constrains $\beta(t)$ according to a quadratic expression that, after some simplifications, may be rewritten as
\begin{equation}
	\label{eq:beta_crit_quadratic}
	-(k_1^2-1)^2\beta^2(t) + (2k_1^4 + 4k_1k_2 +2)\beta(t) - (k_1^2+1)^2>0.
\end{equation}
Through simple algebraic steps, one finds that the roots of the left-hand side of Eq.~\eqref{eq:beta_crit_quadratic} are
\begin{equation}
	\label{eq:beta_crit_roots}
	\begin{cases}
		\beta_{1,2} = \dfrac{k_1^4 + 2k_1k_2 +1 \pm 2 \sqrt{(k_1 k_2 +1)(k_1k_2 + k_1^4)}}{(k_1^2-1)^2}, & k_1\neq 1,\\
		\beta_1 = \dfrac{1}{1+k_2}, &k_1=1,
	\end{cases}
\end{equation}
which exist for all $k_1,k_2>0$.
Given that the quadratic coefficient in Eq.~\eqref{eq:beta_crit_quadratic} is negative, Sylvester's criterion is thus satisfied for $\beta_1<\beta(t)<\beta_2$ whenever $k_1\neq 1$, assuming $\beta_2>\beta_1$, without loss of generality -- i.e. letting $\beta_1$ identify the negative square-root option. Alternatively, if $k_1=1$, Sylvester's criterion is satisfied for $\beta(t)>\beta_1$, since the linear coefficient in Eq.~\eqref{eq:beta_crit_quadratic} is positive. As shown in \ref{ap:sylvesters_crit}, it is easy to verify that $0<\beta_1<1$ in both cases and that $\beta_2>1$, if $k_1\neq1$. Since we are interested in a relief that lies between $[0,1]$, one concludes that a sufficient condition to guarantee stability under the Lyapunov candidate function from Eq.~\eqref{eq:base_lyap} is that $\beta_\text{crit}<\beta(t)\leq 1$, where $\beta_\text{crit}=\beta_1$ identifies the smallest available root from Eq.~\eqref{eq:beta_crit_roots}, given $k_1,k_2$.

In Fig.~\ref{fig:beta_crit}, the value of $\beta_\text{crit}$ is provided over a region of interest in $k_1$ and $k_2$. The analysis of this figure highlights how, from the point of ensuring stability under the proposed Lyapunov function from Eq.~\eqref{eq:base_lyap}, the linear term of the control law may be relieved significantly with an appropriate choice of control gains. Naturally, as the gains themselves alter the control command, it is not necessarily evident that a larger relief translates directly into a decrease of the final command's magnitude. We leave this discussion for a later point, and formally establish the results obtained as follows, introducing a minimum allowed relief factor, $\beta_\text{min}$, to avoid the limit condition of ${\beta(t)=\beta_\text{crit}}$.

\begin{figure}[htpb]
	\centering
	\includegraphics[width=0.9\textwidth]{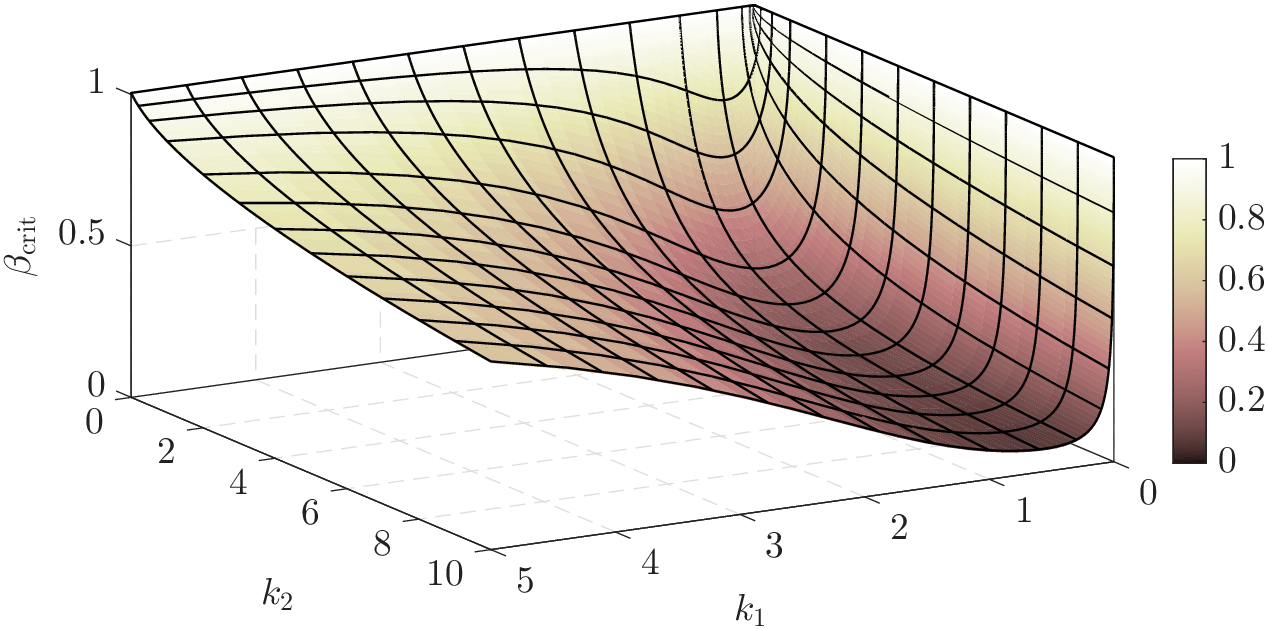}
	\caption{Critical relief, $\beta_\text{crit}$, over $k_1=(0,5]$ and $k_2=(0,10]$.}
	\label{fig:beta_crit}
\end{figure}

\begin{proposition}
	\label{prop:relieved_stability}
	Consider the conditions of Proposition~\ref{prop:base_backstepping}, but under the relieved control law from Eq.~\eqref{eq:relieved_law}, where $\mathbf{K}_1=k_1\mathbf{I}_3$ and $\mathbf{K}_2=k_2\mathbf{I_3}$, with constant $k_1,k_2>0$, and $\beta(t)\geq \beta_\text{min},~\forall t\geq0$, where $\beta_\text{min}\in (\beta_\text{crit},1]$ is a constant and $\beta_\text{crit}$ is the smallest root of Eq.~\eqref{eq:beta_crit_roots}. Then, the application of the relieved control law from Eq.~\eqref{eq:relieved_law} guarantees almost global uniform exponential stability of the origin of the error dynamics in Eq.~\eqref{eq:HFEM_variational}.
\end{proposition}
\begin{proof}
	The proof is completely identical to that of Proposition~\ref{prop:base_backstepping}, changing only Eq.~\eqref{eq:dot_V_aux} to
	\begin{equation}
		\label{eq:dot_V_aux2}
		\dot{V}(\mathbf{z}(t)) \leq -\frac{1}{2}\lambda_\text{min}(\mathbf{U}^*)\norm{\mathbf{z}(t)}^2,~\forall \mathbf{z}(t)\in \mathcal{D}(t),
	\end{equation}
	where $\mathbf{U}^*:=\mathbf{U}\vert_{\beta(t)=\beta_\text{min}}$, such that $c=\frac{1}{2}\lambda_\text{min}(\mathbf{U}^*)$. Under the assumptions of Proposition~\ref{prop:relieved_stability}, one has $c>0$ since $\mathbf{U}^*\succ 0$, as previously discussed.
\end{proof}
\begin{corollary}[of Proposition~\ref{prop:relieved_stability}]
	\label{cor:ell_restriction_2}
	Consider the conditions of Proposition~\ref{prop:relieved_stability} and the ellipsoid $\mathcal{E}$, from Eq.~\eqref{eq:ell_set}. Then, all trajectories beginning at $\mathbf{z}_0 \in \mathcal{E}$ converge exponentially fast to the origin of Eq.~\eqref{eq:HFEM_variational} over $t\in[0,t_s]$.
\end{corollary}
\begin{proof}
	The proof is analogous to that of Corollary~\ref{cor:ell_restriction}, considering the modifications introduced in Proposition~\ref{prop:relieved_stability}.
\end{proof}

Under the guarantees of Proposition~\ref{prop:relieved_stability}, it is possible to bound the evolution of the deviation magnitude through the proposed Lyapunov function. In particular,  Eq.~\eqref{eq:X_aux} may be reorganized to write, at each instant,
\begin{equation*}
	\norm{\mathbf{z}(t)}^2 \geq \frac{2V(\mathbf{z}(t))}{\lambda_\text{max}(\mathbf{X})}.
\end{equation*}
In that case, substitution into Eq.~\eqref{eq:dot_V_aux2} results in
\begin{equation}
	\label{eq:V_dot_aux}
	\dot{V}(\mathbf{z}(t)) \leq - 2\theta V(\mathbf{z}(t)), \quad \text{where} \quad  \theta = \frac{\lambda_\text{min}(\mathbf{U}^*)}{2\lambda_\text{max}(\mathbf{X})},
\end{equation}
which may be integrated to obtain an upper bound on the Lyapunov function itself with time, i.e.
\begin{equation*}
	V(\mathbf{z}(t)) \leq V(\mathbf{z}_0) e^{-2\theta t}.
\end{equation*}
Recalling the result from Eq.~\eqref{eq:X_aux}, one thus establishes
\begin{equation*}
	\frac{1}{2}\lambda_\text{min}(\mathbf{X})\norm{\mathbf{z}(t)}^2 \leq V(\mathbf{z}(t)) \leq V(\mathbf{z}_0) e^{-2\theta t}\leq \frac{1}{2}\lambda_\text{max}(\mathbf{X}) \norm{\mathbf{z}_0}^2 e^{-2\theta t},
\end{equation*}
which, after reorganizing and applying the square-root to both sides, yields an upper bound for the deviation at each instant,
\begin{equation}
	\label{eq:UES_z_bound}
	\norm{\mathbf{z}(t)} \leq \rho \norm{\mathbf{z}_0}e^{-\theta t}, \quad \text{where} \quad \rho=\sqrt{\frac{\lambda_\text{max}(\mathbf{X})}{\lambda_\text{min}(\mathbf{X})}}.
\end{equation}

Under the assumption of exponential stability, we move to ensuring that the saturation threshold is such that it is possible to find $\beta(t)\geq\beta_\text{min}>\beta_\text{crit}$ to meet $\norm{\mathbf{u}(t)}\leq u_\text{sat}$, for all $t\geq 0$. This may, at first, seem counter-intuitive, given how $u_\text{sat}$ is typically an established value that depends on the particular actuator system selected, rather than being a design variable. However, since $\beta_\text{crit}$ depends on the control gains selected, ensuring $\beta(t) \geq \beta_\text{min}>\beta_\text{crit}$ is better interpreted as a problem of compatibility between $u_\text{sat}$, $k_1$, and $k_2$. With this in mind, focusing on the minimum value of $u_\text{sat}$ allows one to separate the stability guarantees from those in terms of saturation. Later, this allows us to formally show that both concerns may be met simultaneously.

To proceed, we note that, at each instant, the triangle inequality provides an upper bound for the magnitude of the actuation command, i.e.,
\begin{equation*}
	\norm{\mathbf{u}(t)} = \norm{\beta(t)\mathbf{K}\mathbf{z}(t) + \mathbf{f}_a(t,\mathbf{z}_1(t))} \leq \beta \norm{\mathbf{K}\mathbf{z}(t)} + \norm{\mathbf{f}_a(t,\mathbf{z}_1(t))}.
\end{equation*}
Thus, in order to satisfy $\norm{\mathbf{u}(t)}\leq u_\text{sat}, ~\forall t\geq 0$, it is sufficient for
\begin{equation}
	\label{eq:ineq_u_sat_aux}
	 u_\text{sat} \geq \beta(t) \norm{\mathbf{K}\mathbf{z}(t)}+ \norm{\mathbf{f}_a(t,\mathbf{z}_1(t))},~\forall t\geq 0.
\end{equation}
At this point, recall that $\beta(t)$ is a design variable and thus, in the limit, one may take $\beta(t)\to\beta_\text{min}$ if need be, without breaking the assumptions necessary for exponential stability provided by Proposition~\ref{prop:relieved_stability}. This choice evidently provides the most lenient version of Eq.~\eqref{eq:ineq_u_sat_aux}, i.e.,
\begin{equation}
	\label{eq:u_sat_aux}
	u_\text{sat} \geq \beta_\text{min} \norm{\mathbf{K}\mathbf{z}(t)}+ \norm{\mathbf{f}_a(t,\mathbf{z}_1(t))}.
\end{equation}
Generally, there is no way to predict how exactly the deviation vector will evolve over time, given a set of initial conditions. Nonetheless, the exponential stability guarantees previously attained bound its magnitude according to Eq.~\eqref{eq:UES_z_bound},  given an initial value -- which we may assume to be a mission parameter. This is evidently relevant for bounding the contributions that make up the right-hand side of Eq.~\eqref{eq:u_sat_aux}, which is trivial to do for the linear term. Fortunately, it is also possible to find an upper bound for $\norm{\mathbf{f}_a(t,\mathbf{z}_1(t))}$ with respect to $\norm{\mathbf{z}_0}$, such that a sufficient condition to meet Eq.~\eqref{eq:u_sat_aux} is given by
\begin{equation}
	\label{eq:u_sat_ineq}
	u_\text{sat} \geq \beta_\text{min}\ell \rho \norm{\mathbf{z}_0} + \phi(\norm{\mathbf{z}_0}),
\end{equation}
where $\ell:=\sigma_\text{max}(\mathbf{K}) =\sqrt{(k_1+k_2)^2+(1+k_1k_2)^2}>1$ is the maximum singular value of $\mathbf{K}$ and $\phi(\norm{\mathbf{z}_0})$ is the upper bound on the acceleration error, i.e. $\norm{\mathbf{f}_a(t,\mathbf{z}_1(t))}$. As detailed in \ref{ap:f_a_bound}, $\phi(\norm{\mathbf{z}_0})$ is a nonlinear expression in $\norm{\mathbf{z}_0}$, which depends on $k_1$, $k_2$, and $\beta_\text{min}$, but also on the target trajectory. Moreover, its validity assumes that the initial deviation satisfies $\mathbf{z}_0\in\mathcal{B}$, where
\begin{equation}
	\label{eq:deviation_constraint_set}
	\mathcal{B}=\left\{\mathbf{x} \in \mathbb{R}^6: \norm{\mathbf{x}} < \frac{r}{\rho}\right\},
\end{equation}
recalling the definition of $r$ from the ellipsoid in Eq.~\eqref{eq:ell_set}. As explored in \ref{ap:f_a_bound}, satisfying $\mathbf{z}_0 \in \mathcal{B}$ implies $\mathbf{z}_0 \in \mathcal{E}$, such that exponential convergence holds, following Corollary~\ref{cor:ell_restriction_2}. Alternatively, given an initial deviation with $\norm{\mathbf{z}_0}<r$, Eq.~\eqref{eq:deviation_constraint_set} may be interpreted as a restriction to $k_1$, through $\rho$.

We provide some additional comments on Eq.~\eqref{eq:u_sat_ineq}. Firstly, in line with the previous discussion, the saturation threshold should be larger than the upper bound of the nonlinear acceleration error, $\phi(\norm{\mathbf{z}_0})$. In fact, we now see that it is further constrained by the linear contribution $\beta_\text{min}\ell\rho\norm{\mathbf{z}_0}$. Note that both contributions depend on the control gains and, moreover, on $\beta_\text{min}$, which up to this point has been interpreted as a design variable, subject only to $\beta_\text{min}>\beta_\text{crit}$, for the purpose of ensuring exponential stability.
From the perspective of the linear term, taking $\beta_\text{min}$ as close to $\beta_\text{crit}$ as possible clearly leads to the most lenient choice. However, this would bring the minimum exponential convergence rate, $\theta$, to zero, following Eq.~\eqref{eq:V_dot_aux}. Given that the nonlinear acceleration error bound may benefit from stronger exponential convergence guarantees, as approached in \ref{ap:f_a_bound}, perhaps it is beneficial to consider larger values of $\beta_\text{min}$. Clearly, it is non-trivial to predict what the best choice for the minimum relief factor is. In fact, this issue may only be fully discussed when particular test cases are analyzed and $\mathbf{w}_j(t),~j=1,\dots,P,$ is precisely known. For this reason, further discussion is diverted to Section~\ref{sec:res}.

For now, the following statement may be presented.
\begin{proposition}
	\label{prop:saturation_stability}
	Consider the conditions of Proposition~\ref{prop:relieved_stability} and let \linebreak$\mathbf{z}_0\in\mathcal{B}$, with $\mathcal{B}$ from Eq.~\eqref{eq:deviation_constraint_set}, denote the deviation from the target trajectory at $t=0$. Let also $u_\text{sat}$, $k_1$, and $k_2$ be such that Eq.~\eqref{eq:u_sat_ineq} is satisfied for $t\in[0,t_s]$, $t_s>0$, with $\beta_\text{min}>\beta_\text{crit}$. Then, it is possible to choose $\beta(t)$, at each instant $t\in[0,t_s]$, such that the control law from Eq.~\eqref{eq:relieved_law} guarantees, during this time interval: \textit{(i)} exponentially fast convergence towards the origin of the error dynamics in Eq.~\eqref{eq:HFEM_variational}; and \textit{(ii)} that $\norm{\mathbf{u}(t)}\leq u_\text{sat}$.
\end{proposition}
\begin{proof}
	The proof for the first statement follows directly from Corollary~\ref{cor:ell_restriction_2}, since $\mathbf{z}_0\in\mathcal{B}\implies \mathbf{z}_0\in \mathcal{E}$, as discussed in \ref{ap:f_a_bound}. Thus, exponential stability is guaranteed over $t\in[0,t_s]$ for any relief satisfying $\beta(t)\geq \beta_\text{min}$ during this time interval, with $\beta_\text{min}\in(\beta_\text{crit},1]$. This validates the upper bound on $\norm{\mathbf{f}_a(t,\mathbf{z}_1(t))}$, also presented in \ref{ap:f_a_bound}.
	
	As previously demonstrated, if Eq.~\eqref{eq:u_sat_ineq} is satisfied over $t\in[0,t_s]$, the minimum relief factor necessary to meet $\norm{\mathbf{u}(t)}\leq u_\text{sat}$ during this period also satisfies ${\beta(t)\geq\beta_\text{min}}$. Hence, there is always at least one choice for $\beta(t),~\forall t\in[0,t_s]$ that ensures $\norm{\mathbf{u}(t)}\leq u_\text{sat}$, without a sacrifice to the exponential stability guarantees.
\end{proof}

The relief factor necessary to meet $\norm{\mathbf{u}(t)}=u_\text{sat}$ when saturation occurs may be calculated, at each instant, by solving the quadratic equation in $\beta(t)$,
\begin{equation}
	\label{eq:beta_calculation}
	u^2_\text{sat} = \norm{\mathbf{u}(t)}^2 \Longleftrightarrow h(\beta(t)):=a \beta^2(t) + b \beta(t) + c = 0, 
\end{equation}
where $a=\norm{\mathbf{K}\mathbf{z}}^2$, $b=2\mathbf{f}_a^T\mathbf{z}$, and $c = \norm{\mathbf{f}_a}^2 - u^2_\text{sat}$, omitting dependencies. Under Proposition~\ref{prop:saturation_stability}, the assumption that Eq.~\eqref{eq:u_sat_ineq} is satisfied guarantees Eq.~\eqref{eq:ineq_u_sat_aux} is also satisfied, such that $h(\beta_\text{min})=\norm{\mathbf{u}(t)}_{\beta(t)=\beta_\text{min}}^2-u^2_\text{sat}\leq0$.~Moreover, $h(1)= \norm{\mathbf{u}(t)}_{\beta(t)=1}^2-u^2_\text{sat}>0$, since saturation is considered only when the unrelieved control law (i.e., with $\beta(t)=1$) saturates. Since $a>0, \forall \mathbf{z}\neq \mathbf{0}$, it must be the case that Eq.~\eqref{eq:beta_calculation} always has a unique\footnote{Neglecting the limit case where $h(\beta_\text{min})=0$, for which $\beta_\text{min}$ is also a root.} solution $\beta(t)\in[\beta_\text{min},1]$.

\section{Numerical Results}
\label{sec:res}

Having established the formalities of the proposed nonlinear control law, with and without the inclusion of actuation saturation, we move to evaluating its performance under various application scenarios, considering typical operational challenges and external perturbations. To this end, we initially assess the adequacy of the base control law from Eq.~\eqref{eq:base_law} at relevant station-keeping tasks, under operational errors and constraints but subject to no external perturbations. Through a Monte Carlo analysis of relevant benchmarks, we compare its performance with continuous and impulsive alternatives from the state-of-the art. Then, the effects of external perturbations are evaluated and stacked against the unperturbed baseline. Finally, actuation saturation is formally included and the final control law from Eq.~\eqref{eq:relieved_law} is validated numerically, offering insights for the optimal selection of control gains and insertion point.

\subsection{Preliminaries}
\label{subsec:res_preliminaries}
The controller will consider the HFEM dynamics in Eq.~\eqref{eq:HFEM_variational} under the point-mass gravitational influence of the Earth, Moon, and Sun, such that $P=3$.  To better replicate a real mission, several operational challenges to the navigation and control of the spacecraft are considered, following standard practice from the literature \cite{pavlak2012Strategy,qi2019ContinuousThrust,zhang2022StationKeepingHFEM}. At the initial instant, an orbital insertion error is pondered, such that the spacecraft begins its mission with a deviation relative to the target trajectory. This error is defined randomly, considering normal distributions with zero mean and covariance $\sigma_{I,r}$ and $\sigma_{I,v}$ in position and velocity, respectively, along each direction. Moreover, we assume that there are random navigation errors that affect the measurement of $\mathbf{z}_1(t)$ and $\mathbf{z}_2(t)$, which are also modeled through normal distributions with zero mean and covariance $\sigma_{N,r}$ and $\sigma_{N,v}$, respectively. Similarly, the application of a given control command is considered to have a normally-distributed error with zero mean and covariance $\sigma_{C}$, defined as a percentage of the total command magnitude. In addition, a minimum threshold to the control command norm, $u_\text{min}$, is prescribed to model limitations of electrical propulsion systems. For now, no upper saturation threshold is established, as this will be considered later when the control law with actuation saturation is tested. Finally, given typical operational constraints, measurements of the spacecraft's position and velocity are assumed to be available only once every $T_M$ days.

Given the measurement interval $T_M$, the on-board computer must propagate the dynamics according to its internal model between successive measurements. We assume this model operates on the same basis as the controller, such that the dynamics are propagated considering only the influence of the Earth, Moon, and Sun, modeled as point-masses. Fig.~\ref{fig:nav_errors_example} provides an illustrative example of the predicted and true state evolution for a spacecraft following a given target trajectory, after a large insertion error. At the end of each predicted segment, resulting from the integration of the dynamics by the on-board computer from the last available measurement, the calculated state of the spacecraft is expected to differ from the true state. The main contributors to this discrepancy include the navigation errors superimposed on said measurement, limitations of the dynamical model considered on-board, control errors, and external disturbances. After each measurement interval, the internal spacecraft state is updated via new measurements and the controller acts so as to bring the predicted state closer to the target trajectory.

\begin{figure}[htpb]
	\centering
	\includegraphics[width=\textwidth]{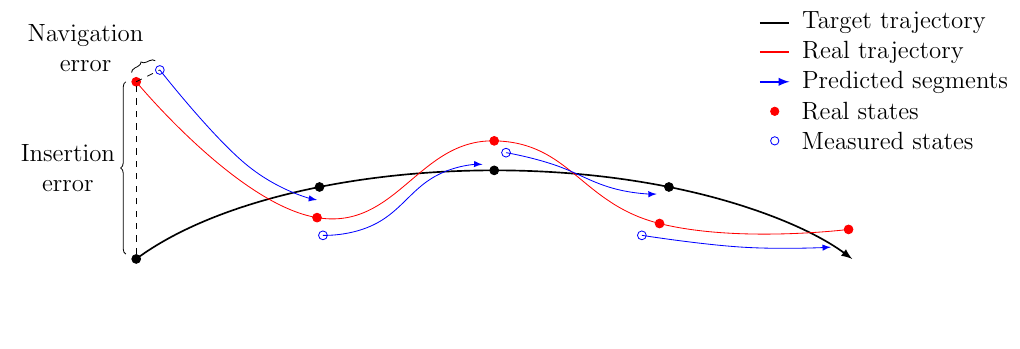}
	\caption{Illustrative example of the real and predicted trajectory of a controlled spacecraft, subject to operational errors and limitations.}
	\label{fig:nav_errors_example}
\end{figure}

To more objectively evaluate the controller performance, we ponder the analysis of two metrics defined as
\begin{equation*}
	E_v = \int_{t_i}^{t_f} \norm{\mathbf{u}(t)} dt \qquad \text{and} \qquad E_e = \int_{t_i}^{t_f} \norm{\mathbf{u}(t)}^2 dt,
\end{equation*}
which measure the total velocity variation enacted by the actuators and the corresponding energy expenditure, respectively, between time instants $t_i$ and $t_f$, in physical units. In addition, we establish
\begin{equation*}
	\text{env}~\norm{\mathbf{z}_1} = \max_{t\in[t_i',t_f]} \norm{\mathbf{z}_1(t)} \qquad \text{and} \qquad 	\text{env}~\norm{\mathbf{z}_2} = \max_{t\in[t_i',t_f]} \norm{\mathbf{z}_2(t)}
\end{equation*}
to assess the upper envelopes of the position and velocity deviation, after the initial insertion error has been (sufficiently) corrected -- which requires an appropriate choice of $t_i'$. Finally, the maximum control command magnitude, $\max \norm{\mathbf{u}}:=\max_{t\in[t_i,t_f]} \norm{\mathbf{u}(t)}$, and total actuator idle time\footnote{In this work, the controller is idle whenever the command's magnitude is below $u_\text{min}$.}, $T_\text{idle}$, are also recorded. Given that the navigation and control errors are defined randomly, all established metrics are evaluated statistically through a Monte Carlo analysis, considering $N_s$ simulations under similar conditions.

The true dynamics of the spacecraft are propagated through an $8$-th order adaptive step Runge-Kutta-Verner method with absolute and relative tolerances set to $10^{-12}$. The exact dynamical model used will depend on whether or not external perturbations are included. For speed, the dynamics considered internally by the on-board computer are propagated through an $8$-th order adaptive step Runge-Kutta-Fehlberg method with more lenient tolerances of $10^{-8}$. The simulations carried out in this work were implemented in Python, making recourse of TU Delft's Astrodynamics Toolbox (Tudat) \cite{tudatspace}, which provides access to the above mentioned integration schemes and serves as the interface with NASA's SPICE planetary ephemeris data for the motion of the celestial bodies considered.

An open-source implementation of the nonlinear backstepping control law proposed, which may be tested under the application cases to follow, is available at \url{https://github.com/antoniownunes/NL_SK_mwe}. The reader is invited to consult this reference for a better understanding of the control solution's implementation and the interface with the Tudat package.

\subsection{Target Trajectories}

In this work, we restrict the analysis to station-keeping about libration point QPOs in the HFEM, given their well-established value for cislunar operations and interplanetary missions. In particular, we consider two trajectories from the Northern Halo and Lyapunov families of QPOs about the L2 Lagrange point -- a strategic outpost amply studied in the literature, key for initiatives such as the Artemis Program \cite{folta2010StationKeeping}. Namely, they are direct counterparts to equivalent periodic orbits of the CR3BP, cataloged in \cite{jpl_3bp_orbits}, with the approximate orbital periods of $14.98$ and $15.86$ days, respectively. These trajectories are transitioned into the HFEM over $25$ revolutions, totaling just over $1$ year in both cases, starting from the 1st of January, 2025, without loss of generality. To do so, the approach presented in \cite{Nunes2026TrajectoryDesign} is pursued, considering four patch points per revolution and targeting a relative segment-to-segment continuity tolerance of $10^{-12}$. The exact QPOs to be targeted reflect the complete dynamical model used when testing the control solution, and thus depend on whether or not external perturbations are considered. Nonetheless, their shape is mostly equivalent between the unperturbed and perturbed test cases, and thus the QPOs designed considering only the point-mass influence of the Earth, Moon, and Sun may be interpreted as a baseline. These are presented in Fig.~\ref{fig:nom_QPOs}, in the synodic reference frame centered at the Earth-Moon barycenter, that rotates with these bodies and pulsates to fix their positions to the corresponding locations of the CR3BP, previously discussed in Section~\ref{subsec:cr3bp}. The projections of both QPOs, the transitioned patch points, and the location of the Moon and L2 point are also evidenced.

\begin{figure}[htpb]
	\centering
	\begin{subfigure}{.48\textwidth}
		\centering
		\includegraphics[width=\linewidth]{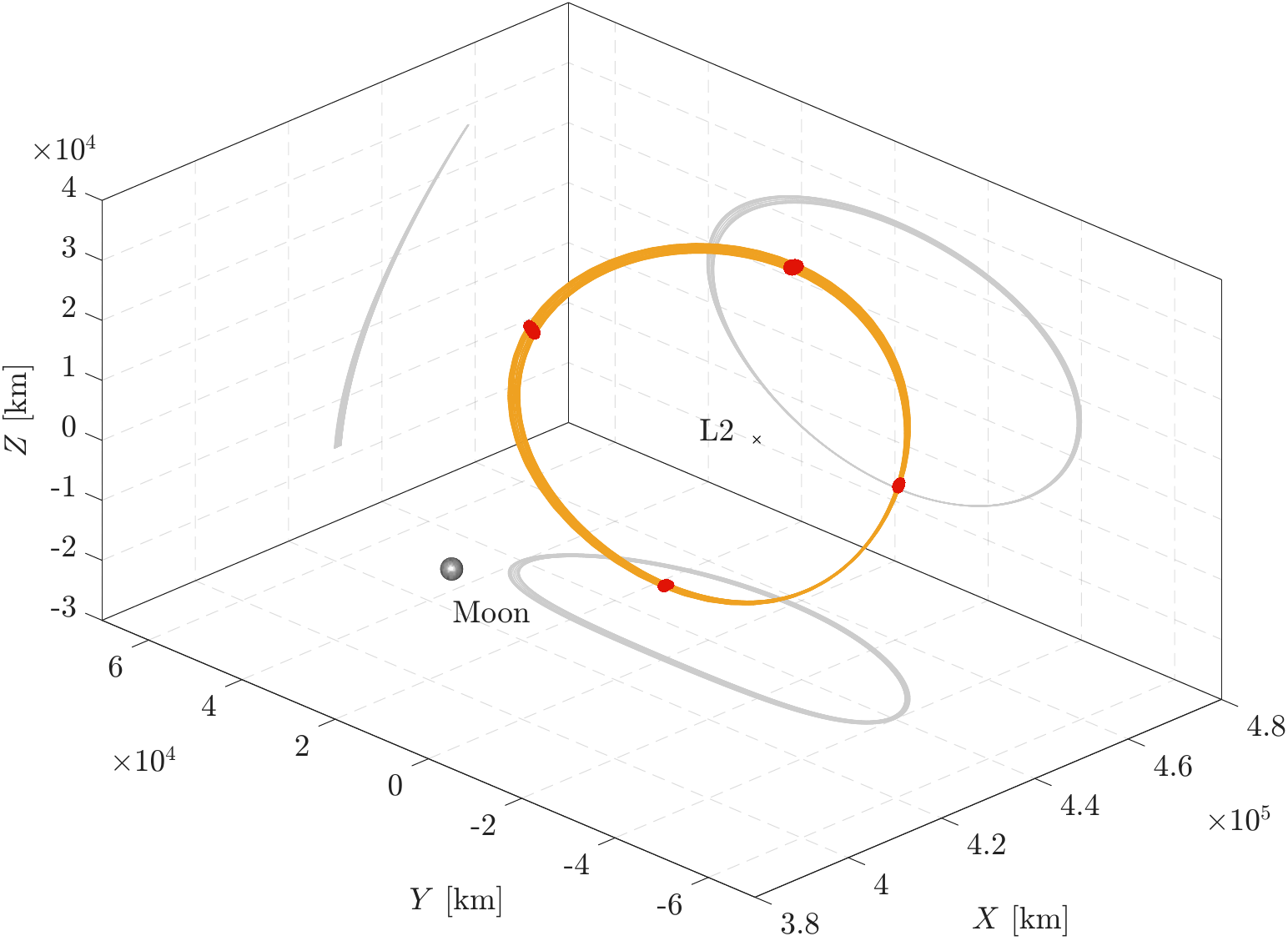}
		\caption{L2 Northern Halo QPO.}
		\label{fig:L2_halo}
	\end{subfigure}%
	\hfill
	\begin{subfigure}{.48\textwidth}
		\centering
		\includegraphics[width=\linewidth]{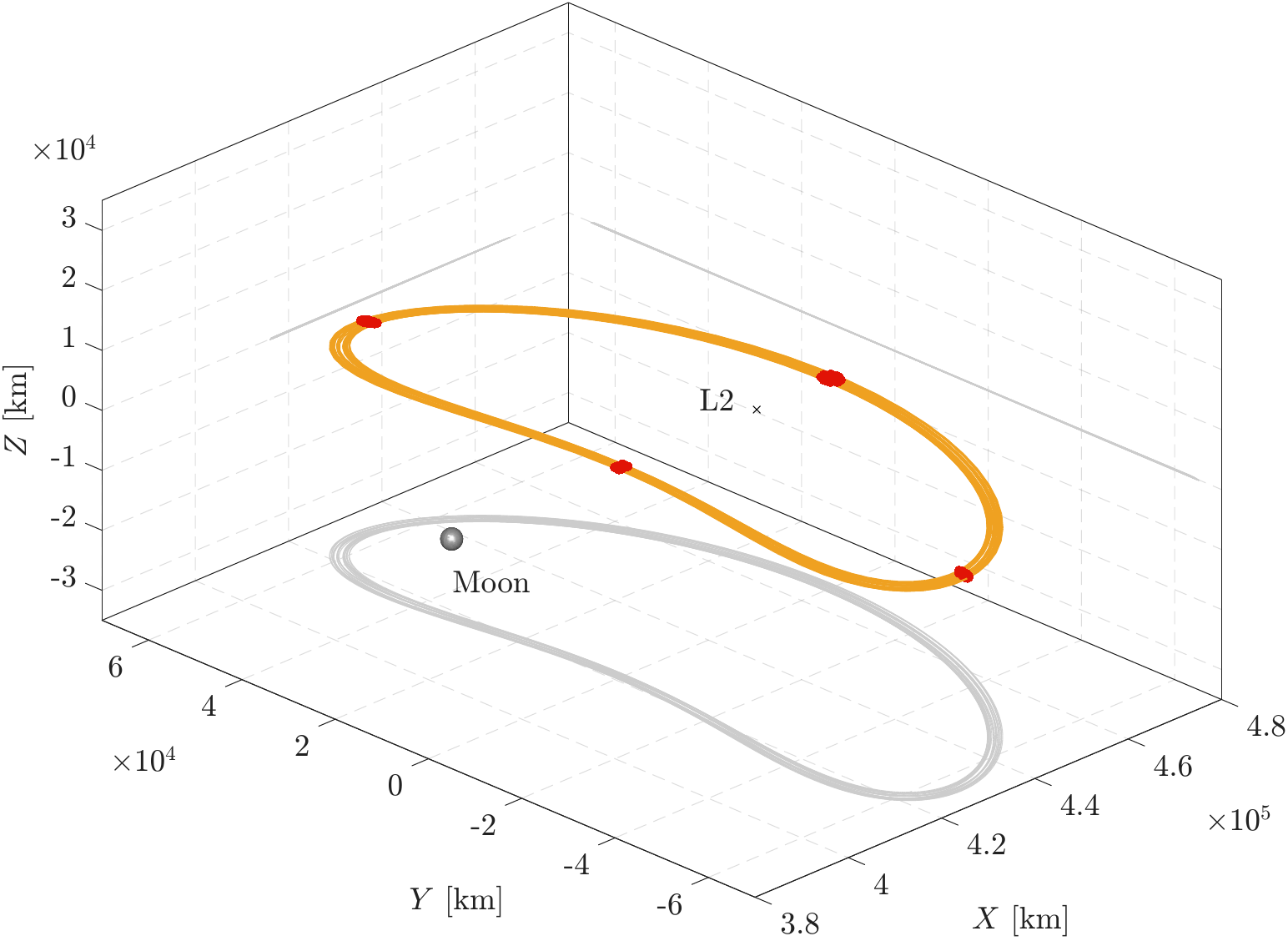}
		\caption{L2 Lyapunov QPO.}
		\label{fig:L2_lyap}
	\end{subfigure}
	\caption{Nominal QPOs in the Earth-Moon synodic reference frame. Designed considering the point-mass influence of the Earth, Moon, and Sun. The Moon is to scale.}
	\label{fig:nom_QPOs}
\end{figure}

Note that the resulting trajectories are similar to those considered in \cite{qi2019ContinuousThrust} and \cite{pavlak2012Strategy}, respectively. For station-keeping, the former considers a continuous LQR implementation, striving for optimality under a linear approximation of the dynamics. The latter presents a two-stage impulsive approach based on multiple-shooting, intertwining control maneuvers with trajectory design for loose reference following. By considering the QPOs in Fig.~\ref{fig:nom_QPOs}, it is possible to compare the station-keeping performance of the proposed approach with both continuous and impulsive alternatives available in the literature. 

\subsection{Base Control Law}

Firstly, the control law with no actuation saturation from Eq.~\eqref{eq:base_law} is employed for station-keeping of the QPOs from Fig.~\ref{fig:nom_QPOs}, so as to evaluate the performance of the base nonlinear backstepping technique. At this point, no external perturbations are considered, such that both the on-board and true dynamics consider only the point-mass influences of the Earth, Moon, and Sun.

In regards to the evaluation of the base control law's performance at station-keeping about the L2 Northern Halo QPO, from Fig.~\ref{fig:L2_halo}, the various operational errors and constraints previously defined in Section~\ref{subsec:res_preliminaries} were selected in accordance with \cite{qi2019ContinuousThrust}, and are collected in Table~\ref{tab:L2_Halo_op_params}. This allows for an almost direct comparison with the results obtained for the continuous LQR alternative from the literature. In particular, a relatively large insertion error in position is considered, to evaluate the adequacy of the proposed controller at eliminating challenging deviations.

\begin{table}[h]
	\centering
	\caption{Operational errors and constraints for the Monte Carlo analysis of station-keeping about the L2 Northern Halo QPO.}
	\label{tab:L2_Halo_op_params}
	\makebox[\linewidth][c]{%
	\begin{tabular}{llllllll}
		\hline
		Parameter & $\sigma_{I,r}$ & $\sigma_{I,v}$ & $\sigma_{N,r}$ & $\sigma_{I,r}$ & $\sigma_C$ & $u_\text{min}$ & $T_M$\\
		\hline
		Value & $100~\si{\kilo \meter}$ & $1~\si{\centi \meter \per \second}$ & $1~\si{\kilo \meter}$ & $1~\si{\centi \meter \per \second}$ & $2\%$ & $0.1~\si{\micro \meter \per \second \squared}$ & $2~$days\\
		\hline
	\end{tabular}
	}
\end{table}

The results of the Monte Carlo trials totaling 100 runs are presented in Table~\ref{tab:MC_L2_Halo_res}, considering the benchmarks previously established, under the choice of $t_i=0$ and ${t_f=365}$ days. The average values and standard deviation of all runs are presented, as denoted by $\bar{\chi}$ and $\sigma_{\chi}$, respectively. The envelope metrics disregard the first 50 days of simulation, i.e. we take $t_i'=50$ days. Later, the adequacy of this selection is confirmed. For further analysis, two sets of trials are conducted, corresponding to the selection of apoapsis and periapsis as the orbital insertion point. The control gains are selected as $k_1=0.5=k_2$, stemming from a thorough preliminary analysis \cite{Nunes2025backstepping}, where a parametric study was carried out to assess the effect of varying $k_1$ and $k_2$ on the system response. A compilation of the system responses resulting from these trials is also presented in Figs.~\ref{fig:MC_L2_Halo_apo} and \ref{fig:MC_L2_Halo_per} for runs beginning at apoapsis and periapsis, respectively. These are given in terms of the magnitudes of the position deviation, velocity deviation, and control command, over the first 100 days of simulation. Note that the results obtained will naturally differ for other control gain choices.



\begin{table}[h]
	\centering
	\caption{Performance metrics from Monte Carlo analysis of station-keeping about the L2 Northern Halo QPO, considering $t_i=0$, $t_i'=50$, and $t_f=365$ days. The base control strategy is employed, taking $k_1=0.5=k_2$.}
	\label{tab:MC_L2_Halo_res}
	\begin{tabular}{ll ll ll l}
		\hline
		& & \multicolumn{2}{l}{Insertion at apoapsis} & \multicolumn{2}{l}{Insertion at periapsis} & \\
		\cmidrule(lr){3-4} \cmidrule(lr){5-6}
		& & $\bar{\chi}$ & $\sigma_{\chi}$ & $\bar{\chi}$ & $\sigma_{\chi}$ & Units\\
		\hline
		\parbox[t]{3mm}{\multirow{6}{*}{\rotatebox[origin=c]{90}{$k_1=0.5=k_2$}}}& $E_v$ & 9.7492 & 1.4101 & 10.852 & 1.7565 & $[\si{\meter\per\second}]$ \\
		& $E_e$ & 9.5794 & 7.5280 & 27.374 & 24.311 & $[\si{\milli \meter\squared \per \second \cubed}]$\\
		& $\text{env}~\norm{\mathbf{z}_1}$ & 19.942 & 3.3587 & 20.371 & 3.5290 & $[\si{\kilo\meter}]$\\
		& $\text{env}~\norm{\mathbf{z}_2}$ &  6.9743 & 1.0458 & 7.1908 & 1.3137 & $[\si{\centi\meter\per\second}]$\\
		& $\max \norm{\mathbf{u}}$ & 3.5943 & 1.8429 & 11.512 & 6.5257 & $[\si{\micro \meter \per \second \squared}]$\\
		& $T_\text{idle}$ & 50.765 & 9.9694 & 50.595 & 9.8931 & $[\text{days}]$\\
		\hline
	\end{tabular}
\end{table}
From the analysis of Table~\ref{tab:MC_L2_Halo_res}, it is possible to conclude that the total velocity variation enacted by the actuators under the proposed base station-keeping strategy is of the same order of magnitude of continuous approaches found in literature. In particular, it is on par with the LQR strategy under comparison \cite{qi2019ContinuousThrust}, where an average yearly cost of $\sim15~\si{\meter\per\second}$ was achieved under similar conditions. The same can be said regarding the controller idle time, which ranges from $\sim30$ to $\sim60$ days in \cite{qi2019ContinuousThrust}, depending on whether a dead-band is included or not. As these figures may vary significantly with the gain choice, we make no comments on perceived improvements in terms of the metrics established. Nonetheless, it is noteworthy to highlight that, despite having no optimality considerations imbued in its design, the nonlinear control law proposed in this work is able to match the performance of the LQR from literature. Ultimately, these metrics serve to underline that capturing the nonlinear dynamics without approximations may outweigh the benefits of linear optimality when deviations are large. In addition, while the remaining values cannot be compared directly, it is reasonable to state that the deviation envelopes and maximum actuation command are within a tolerable threshold, when compared to similar works in literature. Comparison between the results relating to insertion at apoapsis and periapsis reveals, unsurprisingly, that the choice of insertion point yields no significant difference to the final deviation envelopes and controller idle time, but may be meaningful in terms of actuation demands. In particular, the maximum requested command magnitude is nearly tripled when the spacecraft is inserted at periapsis, which is similarly felt in terms of control energy expenditure, $E_e$. The effect on the total velocity variation, $E_v$, is minimal, which we postulate to be a result of dilution over the large time period of analysis.

\begin{figure}[htpb]
	\centering
	\begin{subfigure}{.48\textwidth}
		\centering
		\includegraphics[width=\linewidth]{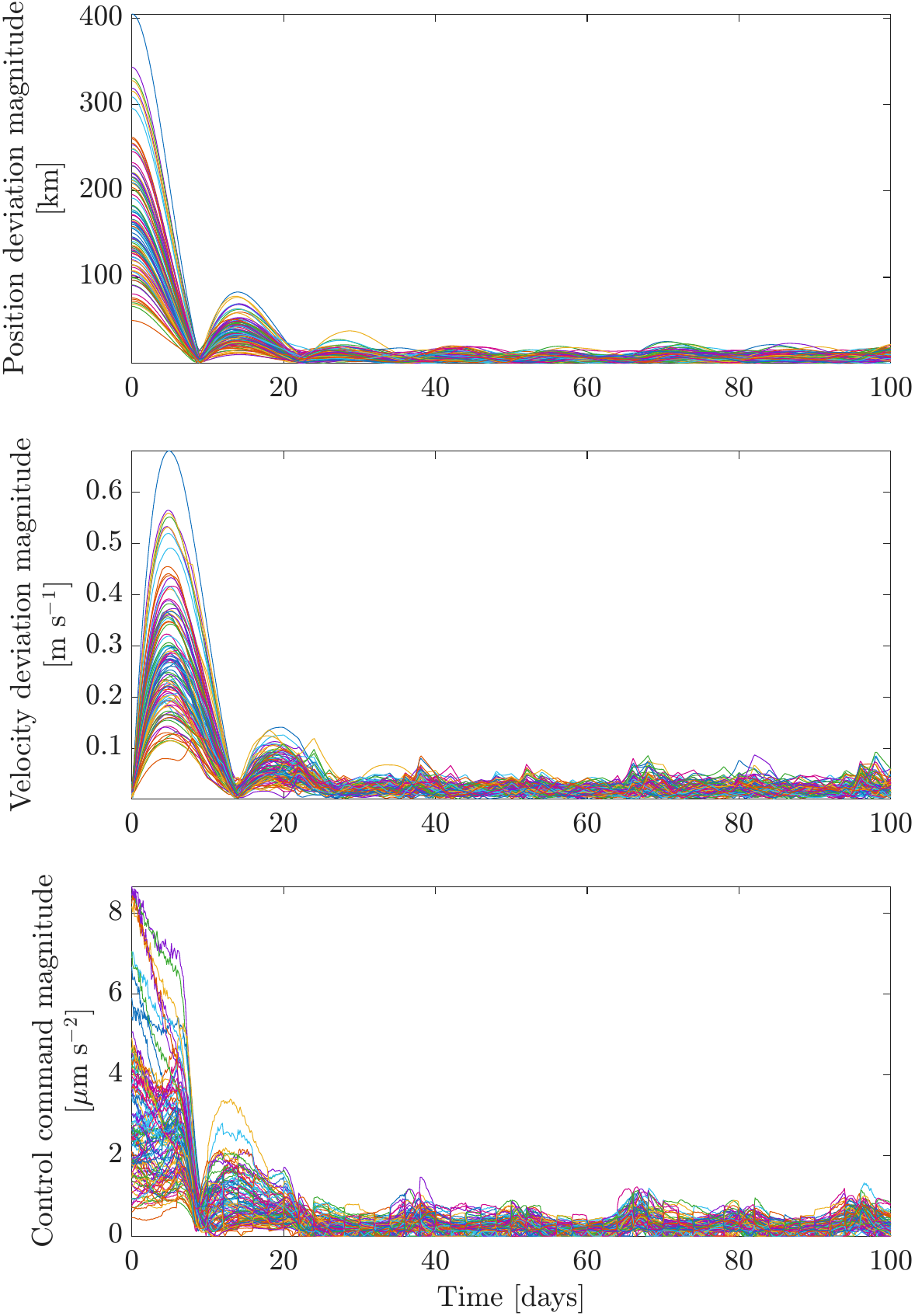}
		\caption{Insertion at apoapsis.}
		\label{fig:MC_L2_Halo_apo}
	\end{subfigure}%
	\hfill
	\begin{subfigure}{.48\textwidth}
		\centering
		\includegraphics[width=\linewidth]{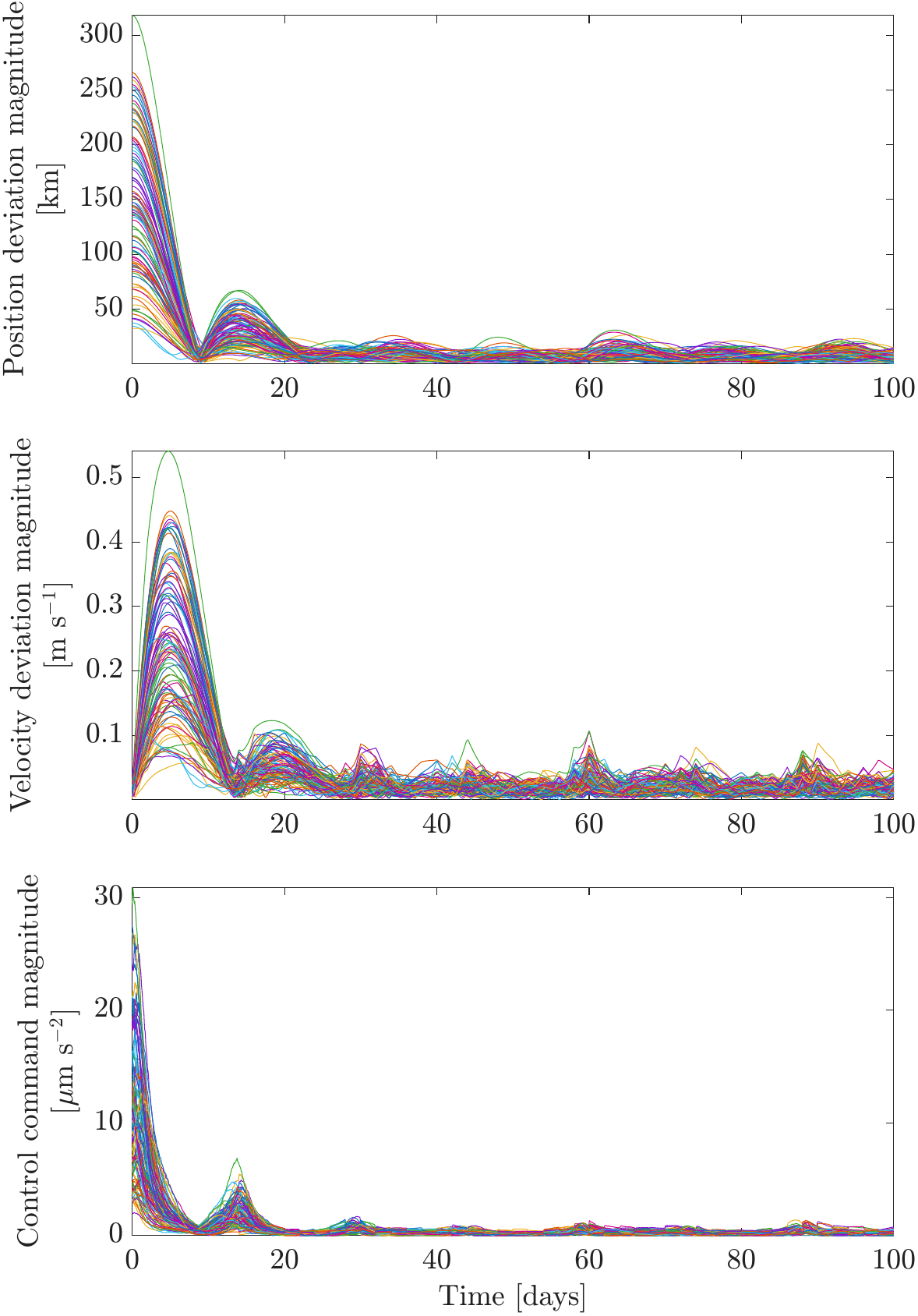}
		\caption{Insertion at periapsis.}
		\label{fig:MC_L2_Halo_per}
	\end{subfigure}
	\caption{System responses of the Monte Carlo analysis of station-keeping about the L2 Northern Halo QPO, over the first $100$ days of simulation, starting at apoapsis and periapsis. The base control strategy is employed, taking $k_1=0.5=k_2$.}
	\label{fig:MC_L2_Halo_traj}
\end{figure}

The observation of Fig.~\ref{fig:MC_L2_Halo_traj} reveals that the choice of insertion point also has little effect on the settling time of the response, following the large initial error, and gives confidence to the selection of $t_i'=50$ days for the calculation of the deviation envelopes in Table~\ref{tab:MC_L2_Halo_res}. It is also possible to confirm the overall larger actuation requests to eliminate the insertion error when the spacecraft starts at periapsis, when compared to apoapsis. In fact, the evolution of the control command magnitude also suggests that the region near periapsis is dynamically more challenging, evidenced by the more demanding control efforts that form oscillations in the plots. These oscillations are not in phase between the two sets of trials due to the different insertion point, but match in terms of orbital placement. No meaningful difference is found in regards to the magnitudes of the position and velocity deviations, which reach similar maximum values and follow a similar overall evolution. In both scenarios, a large overshoot in velocity deviation is visible and may be attributed to the particular gain pair selected in this analysis, as previously approached in \cite{Nunes2025backstepping}. Still, it is possible to confirm that no simulation diverges, attesting to the convergence guarantees postulated in Section~\ref{sec:con} for the proposed controller.

In regards to station-keeping about the L2 Lyapunov QPO, from Fig.~\ref{fig:L2_lyap}, the analysis performed in this work considers the insertion error to be absent. Since the adequacy of the proposed strategy at eliminating large deviations from the target trajectories has already been analyzed in the previous test case, this choice allows one to isolate the yearly performance metrics of station-keeping about the second QPO. In particular, we aim to compare the proposed solution with a well-established impulsive-based alternative, found in \cite{pavlak2012Strategy}.
In this sense, the operational errors and constraints are selected in accordance with the ones considered in \cite{pavlak2012Strategy}, which are summarized in Table~\ref{tab:L2_Lyap_op_params}. However, given the absence of information regarding $u_\text{min}$, due to the impulsive nature of the approach found in the literature, we recover the value used previously, from Table~\ref{tab:MC_L2_Halo_res}. In addition, we evaluate the control performance for three choices of measurement interval, between the previous case of $T_M=2$ days and a more challenging value of $T_M = 8$ days. The latter option matches the interval between the impulsive maneuvers of the shooting strategy employed in \cite{pavlak2012Strategy}, which are applied twice per revolution.

\begin{table}[h]
	\centering
	\caption{Operational errors and constraints for the Monte Carlo analysis of station-keeping about the L2 Lyapunov QPO.}
	\label{tab:L2_Lyap_op_params}
	\begin{tabular}{lllllllll}
		\hline
		Parameter & $\sigma_{I,r}$ & $\sigma_{I,v}$ & $\sigma_{N,r}$ & $\sigma_{I,r}$ & $\sigma_C$ & $u_\text{min}$ & $T_M$\\
		\hline
		Value & $0.0$ & $0.0$ & $1~\si{\kilo \meter}$ & $1~\si{\centi \meter \per \second}$ & $1\%$ & $0.1~\si{\micro \meter \per \second \squared}$ & $2,4,8~$days\\
		\hline
	\end{tabular}
\end{table}

Similarly to the previous test case, the period of analysis is taken as 1 year, leading to $t_i=0$ and $t_f=365$ days for the time instants considered in the evaluation of the performance metrics. Since the insertion error has been neglected, the deviation envelope benchmarks consider $t_i'=0$. For the same reason, it is sufficient to perform the simulations for only one insertion point, which is chosen to be apoapsis. The results stemming from the application of the proposed base control law over 100 trials are presented in Table~\ref{tab:MC_L2_Lyap_res}, for the different choices of $T_M$ considered.

\begin{table}[h]
	\centering
	\caption{Performance metrics from Monte Carlo analysis of station-keeping about the L2 Lyapunov QPO for different navigation intervals, $T_M$, considering $t_i=0=t_i'$ and $t_f=365$ days. The base control strategy is employed, taking $k_1=0.5=k_2$, with insertion at apoapsis.}
	\label{tab:MC_L2_Lyap_res}
	\begin{tabular}{l ll ll ll l}
		\hline
		& \multicolumn{2}{l}{$T_M=2$ days} & \multicolumn{2}{l}{$T_M=4$ days} & \multicolumn{2}{l}{$T_M=8$ days} & \\
		\cmidrule(lr){2-3} \cmidrule(lr){4-5} \cmidrule(lr){6-7}
		& $\bar{\chi}$ & $\sigma_{\chi}$ & $\bar{\chi}$ & $\sigma_{\chi}$  & $\bar{\chi}$ & $\sigma_{\chi}$ & Units\\
		\hline
		$E_v$ & 8.9105 & 0.6805  & 18.818  & 1.6036  & 76.478  & 10.893 & $[\si{\meter\per\second}]$ \\
		$E_e$ & 5.1157  & 0.8481  & 21.643  & 3.7689  & 420.95  & 137.03 & $[\si{\milli \meter\squared \per \second \cubed}]$\\
		$\text{env}~\norm{\mathbf{z}_1}$ & 26.202  & 4.7603  & 56.379  & 12.470  & 234.40  & 56.356 & $[\si{\kilo\meter}]$\\
		$\text{env}~\norm{\mathbf{z}_2}$ & 10.363  & 1.9405  & 21.515  & 5.4819  & 89.733  & 26.170 & $[\si{\centi\meter\per\second}]$\\
		$\max \norm{\mathbf{u}}$ & 2.2384  & 0.3933  & 4.4809  & 0.9133  & 20.106  & 5.7246 & $[\si{\micro \meter \per \second \squared}]$\\
		$T_\text{idle}$ & 62.030  & 9.8522  & 24.874  & 7.3808  & 6.7983  & 3.7609 & $[\text{days}]$\\
		\hline
	\end{tabular}
\end{table}

The results in Table~\ref{tab:MC_L2_Lyap_res} highlight the necessity of sufficiently frequent measurements to update the prediction of the spacecraft state that is fed to the control law. For $T_M = 2$ days, the results reflect the previous analysis carried out for the L2 Northern QPO on nearly all grounds -- control benchmarks see a significant decrease given the absence of a large insertion error, but deviation envelopes and actuator idle time see only slight changes. Even for $T_M= 4$ days, around half of the time interval considered in the impulsive approach covered in \cite{pavlak2012Strategy}, the results remain adequate. In particular, the average yearly velocity cost is on par with the impulsive alternative from literature, where $\sim15-18~\si{\meter\per\second}$ was achieved, under similar conditions, depending on the specific optimization strategy employed. However, for larger measurement intervals, namely $T_M = 8$ days, the accumulation of navigation, control, and numerical errors translate into a significant increase on all accounts of the control and deviation benchmarks, and a decrease in control idle time -- results which may not be meaningfully improved simply through a different gain selection. In such cases, impulsive approaches such as the one in \cite{pavlak2012Strategy}, which optimize a current maneuver by propagating the HFEM dynamics forward in time with very high-accuracy, unsurprisingly take the upper edge. As previously explored, preliminary mission analysis must counterweight this information with the increase in computational demands associated to such an approach, the thrust capabilities of the actuator system employed, and the formal operational guarantees attained. In this sense, if measurements are sufficiently frequent, the proposed nonlinear strategy may prove beneficial because it avoids expensive numerical optimization steps and is able to provide formal guarantees of convergence.

To better interpret the aforementioned results, Fig.~\ref{fig:MC_L2_Lyap} presents the corresponding system responses for the three sets of trials with different values of $T_M$, which may be discerned by their color. In addition, the overall envelopes of each plotted quantity, i.e. the maximum value between all runs for a given $T_M$, at each instant of time, are represented as shaded areas of the same color. The results further confirm the adequacy of the proposed strategy when the measurement interval is small, and its inefficacy when measurements are not sufficiently frequent. As previously discussed, this may be attributed to the accumulation of errors in the position and velocity deviation predictions, which are further exacerbated by the control action itself, which works with incorrect spacecraft state predictions -- creating a snowball effect that grows worse with time. This shortcoming has previously been observed in the literature for other continuous station-keeping strategies and can be ultimately identified as a drawback of such approaches. Potential improvements may be achieved by considering an on-board propagation model with a higher degree of accuracy. However, preliminary analysis has shown that such improvements are very marginal and that the navigation errors and measurement interval clearly dominate the discrepancy between the predicted and true satellite states, which motivate extra control action. Nonetheless, despite the increase in control benchmarks, we remark that all simulations remain in relative proximity to the target trajectory. One may postulate that the formal stability guarantees provided by the control law play an important role in avoiding divergence in these numerical trials.

\begin{figure}[htpb]
	\centering
	\includegraphics[width=\textwidth]{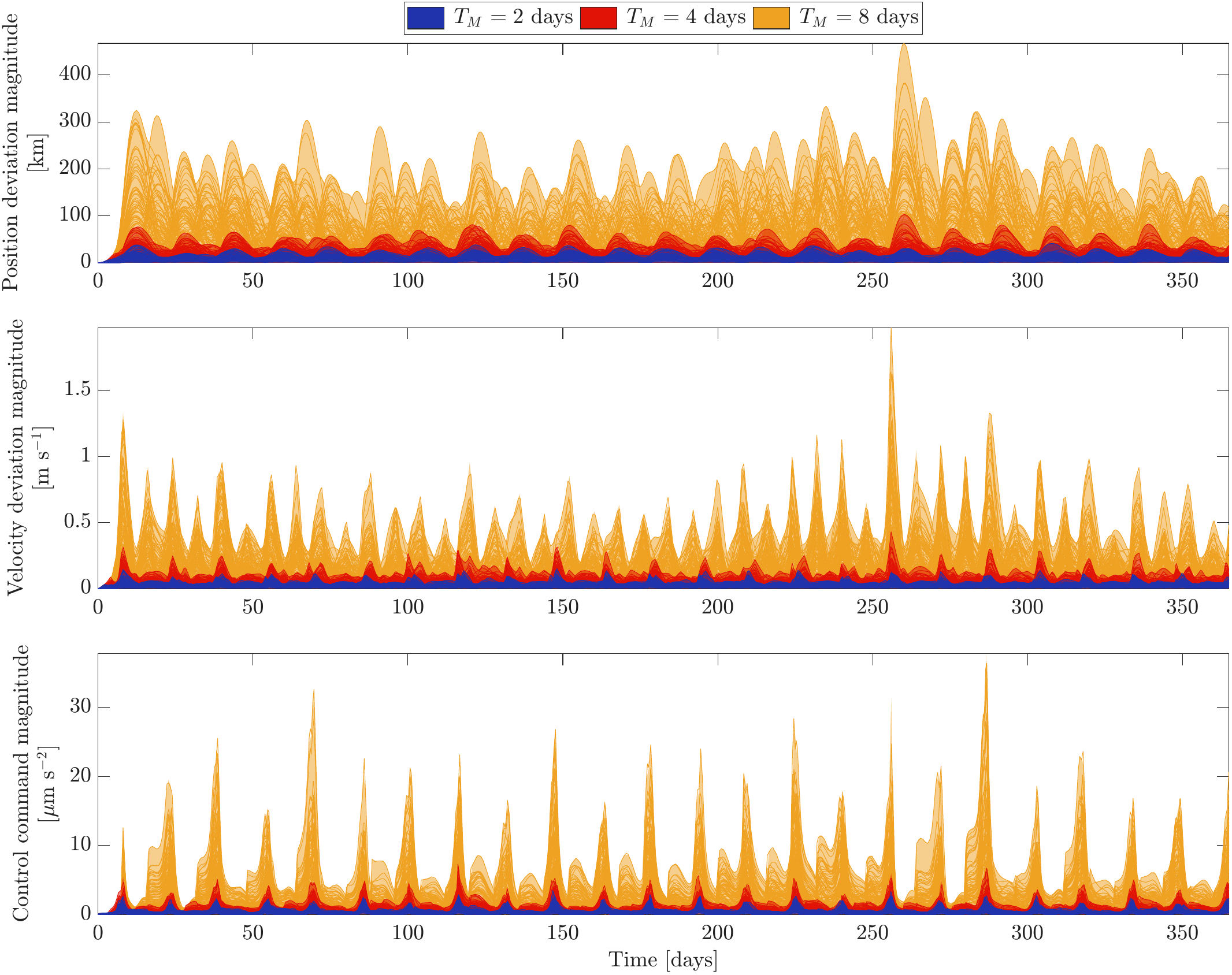}
	\caption{System responses of the Monte Carlo analysis of station-keeping about the L2 Lyapunov QPO, over $1$ year of simulation, starting at apoapsis. The base control strategy is employed, taking $k_1=0.5=k_2$. Shaded regions discern between simulations considering different measurement intervals, $T_M$.}
	\label{fig:MC_L2_Lyap}
\end{figure}

\subsection{External Perturbation Rejection}
\label{subsec:res_external_perts}

In order to further improve the realism of the application, we ponder the inclusion of external perturbations, to be considered at the level of the true dynamics and trajectory design\footnote{In \cite{qi2019ContinuousThrust}, the importance of considering the perturbations when designing the nominal trajectories is underlined. Such efforts avoid unnecessary control action acting against the dynamics. This does not increase the computational cost of the controller, as trajectory design is typically done before deployment, with the highest possible fidelity model.}, but not by the on-board controller. In particular, three contributions are considered: \textit{(i)} spherical harmonics in the gravitational fields of the Earth and Moon; \textit{(ii)} solar radiation pressure (SRP); and \textit{(iii)} gravitational perturbations due to the other relevant bodies in the Solar System. In regards to the inclusion of spherical harmonics in the gravitational fields of the Earth and Moon, we ponder the use of an 8th order model -- a typical standard considered in literature \cite{williams2017NRHO,davis2017Stationkeeping,Spreen2023BaselineNRHO}. For the SRP perturbation, it is necessary to establish additional physical properties of the spacecraft. While this particularizes the analysis presented in this work, it is also the case that it may be interpreted as a proof of concept that may be easily extended to other configurations, for which similar results may be expected. In this work, we assume that the spacecraft has a mass of $m= 500~\si{\kilogram}$, mean area of radiation incidence of $A= 10~\si{\meter \squared}$, and coefficient of reflectivity $C_r=1.8$. The SRP perturbation is calculated with a typical cannonball model, considering possible occlusion due to the Earth and Moon, for improved realism. As for the external body perturbations, only the point-mass contribution of Jupiter is pondered. As will be seen shortly, the remaining bodies of the solar system exert negligible forces on the spacecraft. All perturbing effects are modeled and calculated making recourse of Tudat's capabilities.

For this study, we return to the L2 Northern Halo QPO. In Fig.~\ref{fig:pert_L2_Halo}, the magnitude of the aforementioned perturbations over the path of the target trajectory is presented and compared to the point-mass (PM) influences of the Earth, Moon, and Sun, in a logarithmic scale. This serves to illustrate that, for this particular QPO, the spherical harmonics (SH) in the gravitational field of the Moon are, by far, the most meaningful perturbation to the dynamics -- on par with the point-mass gravitational influence of the Sun. This is followed by the SRP perturbation, around 2 orders of magnitude lower. Note that, for this perturbation's profile, the discontinuities evidence occlusion with respect to the Earth and/or Moon. Finally, the perturbation stemming from the spherical harmonics in Earth's gravitational model and the point-mass contribution of Jupiter present near negligible effects. Additional analysis, here omitted, shows that the contributions of the remaining celestial bodies lie below that of Jupiter and are thus negligible, practically speaking. Moreover, we note that the L2 Lyapunov QPO is affected in a similar manner by the perturbations considered. However, both QPOs occupy a similar region in space and thus the results obtained may differ vastly for other nominal trajectories.

\begin{figure}[htpb]
	\centering
	\includegraphics[width=\textwidth]{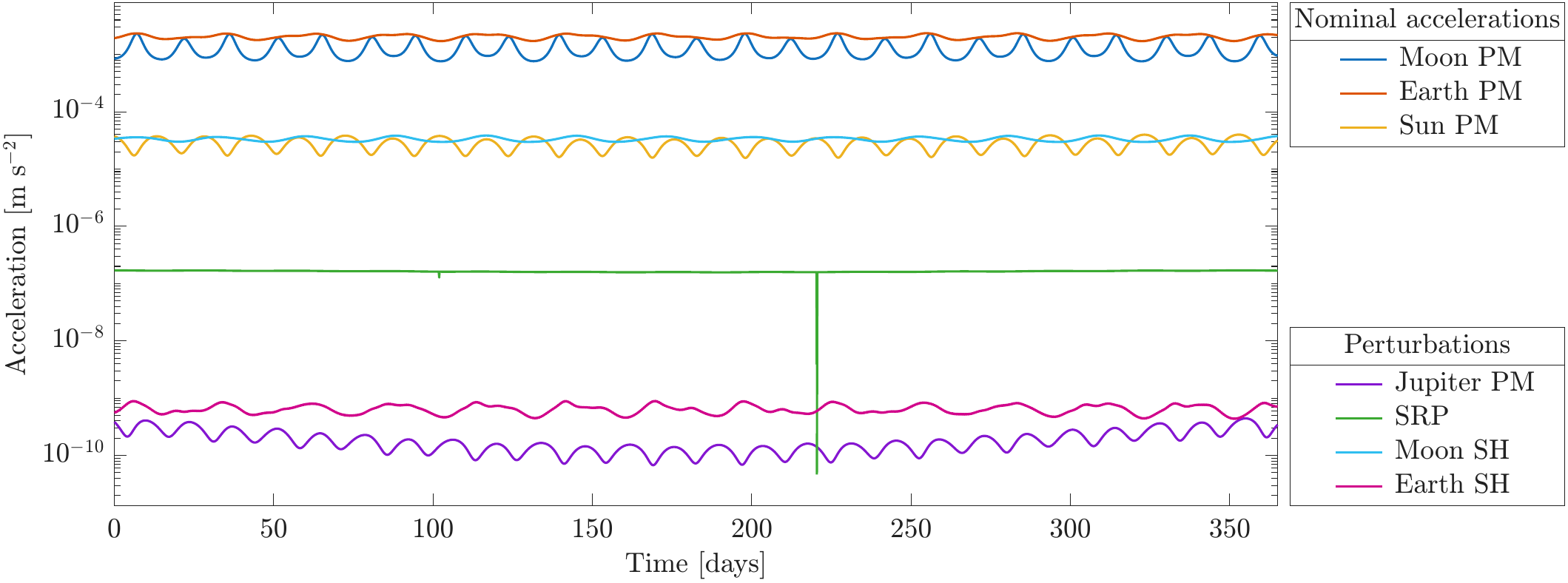}
	\caption{Perturbations felt by a spacecraft following the L2 Northern Halo QPO, over 1 year of simulation, compared to nominal accelerations from the point-mass gravitational contributions of the Earth, Moon, and Sun.}
	\label{fig:pert_L2_Halo}
\end{figure}

To analyze the perturbation rejection capabilities of the base nonlinear control law employed, a Monte Carlo analysis is carried out considering the established external perturbations, recovering also the operational constraints and errors previously presented in Table~\ref{tab:L2_Halo_op_params}. The corresponding results, for missions starting at apoapsis and periapsis are represented in Table~\ref{tab:MC_L2_Halo_pert_res}. We note that, once more, all simulations correctly converge towards a region of proximity to the target trajectory, obeying, on average, the position and velocity envelopes detailed in the table.



\begin{table}[h]
	\centering
	\caption{Performance metrics from Monte Carlo analysis of station-keeping about the L2 Northern Halo QPO under external perturbations, considering $t_i=0$, $t_i'=50$, and $t_f=365$ days. The base control strategy is employed, taking $k_1=0.5=k_2$.}
	\label{tab:MC_L2_Halo_pert_res}
	\begin{tabular}{l ll ll l}
		\hline
		& \multicolumn{2}{l}{Insertion at apoapsis} & \multicolumn{2}{l}{Insertion at periapsis} & \\
		\cmidrule(lr){2-3} \cmidrule(lr){4-5}
		& $\bar{\chi}$ & $\sigma_{\chi}$ & $\bar{\chi}$ & $\sigma_{\chi}$ & Units\\
		\hline
		$E_v$ & 12.513 & 1.4244  & 13.755  & 1.6947  & $[\si{\meter\per\second}]$ \\
		$E_e$ & 12.987  & 8.1150  & 30.123  & 26.160  & $[\si{\milli \meter\squared \per \second \cubed}]$\\
		$\text{env}~\norm{\mathbf{z}_1}$ & 23.117  & 2.6172  & 23.209  & 2.9132  & $[\si{\kilo\meter}]$\\
		$\text{env}~\norm{\mathbf{z}_2}$ & 7.1769  & 1.2194  & 7.1893  & 1.1387  & $[\si{\centi\meter\per\second}]$\\
		$\max \norm{\mathbf{u}}$ & 4.1748  & 2.0565  & 12.552  & 6.5583  & $[\si{\micro \meter \per \second \squared}]$\\
		$T_\text{idle}$ & 32.627  & 7.5591  & 32.430  & 5.4897  & $[\text{days}]$\\
		\hline
	\end{tabular}
\end{table}

The benchmarks obtained highlight that, even under realistic external perturbations, the proposed control law is adequate at eliminating the large insertion error and, moreover, at maintaining the spacecraft in a region of proximity to the target trajectory. To this end, when compared with the results from Table~\ref{tab:MC_L2_Halo_res}, similar deviation envelopes are observed. In terms of actuation, the control effort benchmarks $E_v$ and $E_e$ see a moderate increase of approximately $30\%$. The actuator idle time sees a decrease of similar magnitude in both sets of trials when compared to the unperturbed analysis. This comparison shows that, when external perturbations are introduced, the controller continues to fulfill the proposed tasks but expends moderately more actuation to do so. Nonetheless, the performance observed is still in line with other continuous strategies from literature considering similarly perturbed application scenarios, such as \cite{qi2019ContinuousThrust}, where a yearly cost of $\sim 18~\si{\meter \per \second}$ is achieved under comparable conditions. These results point to a measure of inherent robustness of the proposed nonlinear control law.

\subsection{Control Law with Actuation Saturation}
\label{subsec:res_sat}

Having established the adequacy of the base nonlinear control law at the station-keeping task about two relevant QPOs in the Earth-Moon system, we move to quantifying the effects of formally including an actuation saturation threshold in the controller design. As previously discussed, the approach pursued in this work relies on the coupling of the controller gains and the saturation threshold to formalize exponential stability guarantees in a region of proximity to the target trajectory. In this sense, we recall that the initial deviation from the target orbit must belong to the set $\mathcal{B}$ from Eq.~\eqref{eq:deviation_constraint_set}, which depends on the evolution of the target trajectory itself, namely on its closest approach to the celestial bodies considered over the time frame of analysis. As the target trajectories considered in this work are QPOs, we assume that, for a sufficiently large period of simulation, the number of revolutions that make up the nominal solutions are enough to fully capture their spread in space. For this reason, while the formal guarantees hold only over the time-domain of study, we argue that they may easily be extended for longer periods of time, should a mission require so.

As previously mentioned, the selection of the control gains must be performed harmoniously with the value of the actuation saturation threshold. To better illustrate the coupling between parameters, we study the minimum saturation threshold, $u_\text{sat,min}$, that respects Eq.~\eqref{eq:u_sat_ineq} tightly, given $k_1$ and $k_2$. We focus this analysis on the L2 Northern Halo QPO and an established expected value for the initial deviation, derived from the operational errors from Table~\ref{tab:L2_Halo_op_params}. More specifically, we take $\norm{\mathbf{z}_0}\approx 1 \times 10^{-3}$ as a conservative guess which corresponds to the $99.9\%$ quantile of the insertion error, as determined through statistical analysis. Naturally, as the controller operates subject to operational constraints and errors, there are no formal guarantees on the magnitude of the deviation when new measurements are available and the deviation is updated. This could seemingly put into question the balance between control gains and saturation threshold that was carefully tuned with the expected deviation at insertion. We later discuss possible solutions to this problem but, for now, it is assumed that the insertion error constitutes the maximum deviation from the target orbit. Considering the established value for $\norm{\mathbf{z}_0}$, the maximum $\rho$ for which $\mathbf{z}_0\in\mathcal{B}$ is satisfied imposes $k_1<10.775$. As such, we limit our analysis to the conservative figure of $k_1\leq10$.

We recall that a meaningful decision in regards to the calculation of $u_\text{sat,min}$, from Eq.~\eqref{eq:u_sat_ineq}, lies in the minimum allowed relief factor, subject to $\beta_\text{min}>\beta_\text{crit}$ for each $k_1$ and $k_2$. In order to evaluate the effect of this parameter, we ponder its definition according to an auxiliary parameter, $\eta\in(0,1]$, such that
\begin{equation*}
	\beta_\text{min} = \beta_\text{crit} + \eta (1-\beta_\text{crit}).
\end{equation*}
In this way, $\eta=0$ identifies the critical condition, where the exponential nature of the stability guarantees is lost, as the minimum exponential rate of convergence formally goes to zero -- asymptotic stability is, nonetheless, maintained. Contrarily, increasing $\eta$ ensures the minimum rate of convergence is larger, at the expense of allowing less relief to the control command. In the limit, for $\eta=1$, the relieved control law from Eq.~\eqref{eq:relieved_law} returns to its base form, in Eq.~\eqref{eq:base_law}.

In Fig.~\ref{fig:u_sat_L2_Halo_apo}, the level curves for the ratio $u_\text{sat,min}/\norm{\mathbf{z}_0}$ are provided, for increasing choices of $\eta$, considering insertion at apoapsis. The results obtained are unique to the L2 Northern Halo QPO and the start epoch and time period of analysis selected, which are recalled from previous sections. As such, they may differ greatly for other target trajectories. Still, further analysis shows that changes in start epoch and time-period do not yield significantly different results. As previously discussed, the last plot, for $\eta=1$, may be interpreted as a baseline -- it corresponds to the total application of the original backstepping control law from Eq.~\eqref{eq:base_law} and, thus, the level curves constitute the maximum expected actuation command of the base controller, under the exponential stability guarantees that it provides. On the contrary, the plot for $\eta=0$ represents the limit case where the relief to the original control law is free to reach the maximum theoretical limit for exponential stability.

\begin{figure}[htpb]
	\centering
	\includegraphics[width=\linewidth]{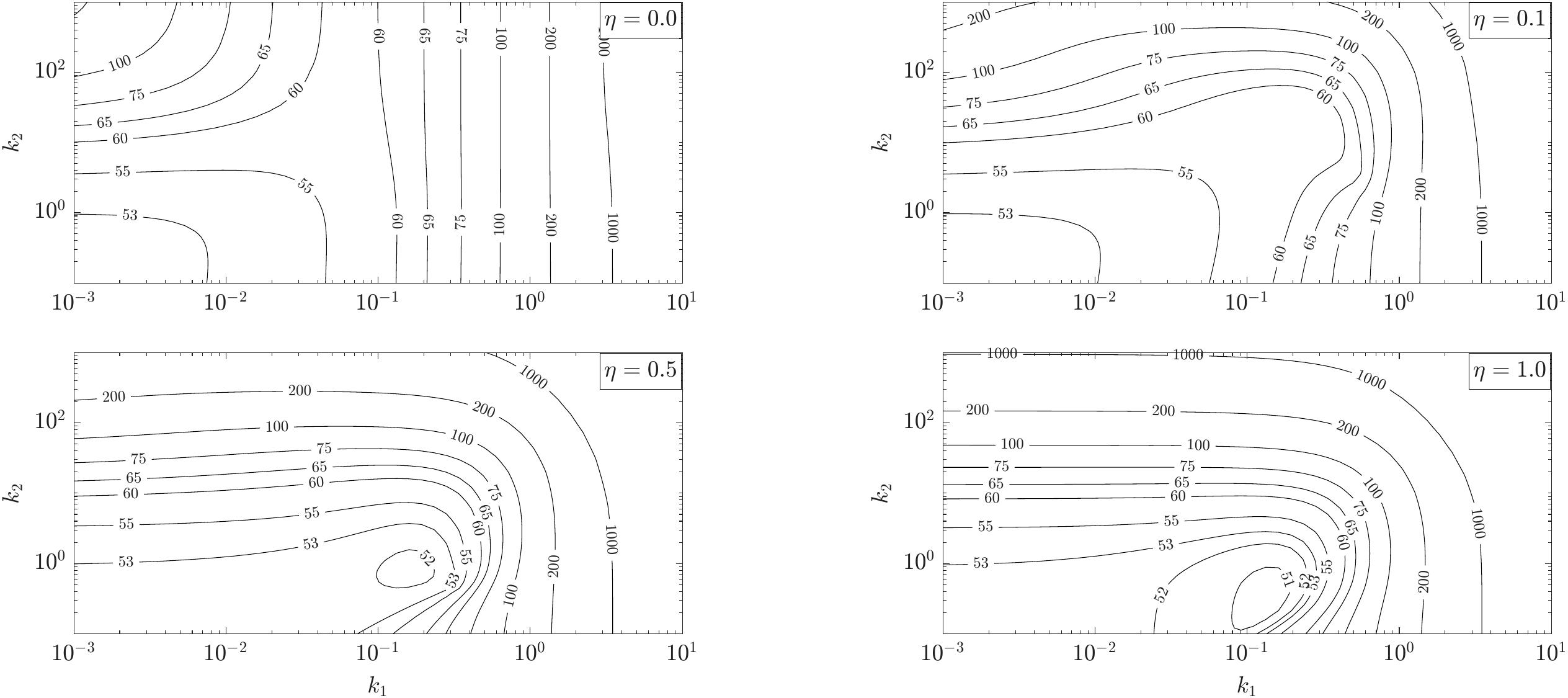}
	\caption{Level curves for the minimum saturation threshold $\frac{u_\text{sat,min}}{\norm{\mathbf{z}_0}}$, for the L2 Northern Halo QPO and $\norm{\mathbf{z}_0}\approx1\times 10^{-3}$, considering insertion at apoapsis.}
	\label{fig:u_sat_L2_Halo_apo}
\end{figure}

One immediate conclusion from the observation of Fig.~\ref{fig:u_sat_L2_Halo_apo} is that, for small $k_2$ and large $k_1$, all plots approach the same profile, irrespective of $\eta$ -- ultimately, $\rho$ explodes and the linear term completely dominates the minimum saturation threshold in Eq.~\eqref{eq:u_sat_ineq}. A similar observation is made regarding small $k_2$ and small $k_1$, however due to the elimination of the linear contribution of the control command -- cases where the minimum saturation threshold is dictated by the nonlinear acceleration term. On the contrary, for moderate choices of $k_1$, it is possible to observe that $\eta$ plays an important role by balancing the opposing nature of the relief tolerance and convergence rate. By direct comparison between the plots for $\eta=1$ and $\eta=0$, it is possible to see that, while the overall minimum value of $u_\text{sat,min}$ is slightly higher for the latter, the region where it remains lower than $100\norm{\mathbf{z}_0}$, for example, is significantly increased. This means that, generally speaking, for the same saturation threshold, the introduction of a relief factor to the control law may provide more lenience in the selection of the control gains. While this is not indicative of a faster response (because, even if larger gains are allowed, the control command is being relieved, after all), it does suggest a larger margin for response shaping through gain selection. Conversely, it might be the case that, for a given gain pair, it is only possible to respect a saturation threshold if a relief to the control law is considered -- see, for example, how $k_2=100$ allows only for $u_\text{sat}>100\norm{\mathbf{z}_0}$ when $\eta=1$, but as low as $u_\text{sat}=60\norm{\mathbf{z}_0}$ when $\eta=0$. As the figure suggests, perhaps a choice of $\eta$ between the limit values of $0$ and $1$ may be better suited, since it provides a balance between the benefits stemming from a larger relief, as in $\eta=0$, and of a faster convergence rate, for $\eta=1$. Clearly, the optimal choice of $\eta$ for each gain pair is generally not evident. We return attention to this matter shortly.

In a similar fashion, Fig.~\ref{fig:u_sat_L2_Halo_per} presents the curve levels for the same function, but considering insertion at periapsis. Once more, the results provided are valid only for the L2 Northern Halo QPO, the selected start epoch, and time period considered. 
We start by highlighting that, in the case of $\eta=0$, the level curves obtained are exactly the same as the previous case from Fig.~\ref{fig:u_sat_L2_Halo_apo}, considering insertion at apoapsis. This is because, when the maximum relief is allowed, the minimum guaranteed rate of convergence goes to zero (the exponential nature of the stability guarantees is lost), such that the nonlinear term that bounds the acceleration error in Eq.~\eqref{eq:u_sat_ineq} is irrespective of the insertion point -- this is further discussed in \ref{ap:f_a_bound}. Similarly, the conclusions previously drawn for very large and very small values of $k_1$ also hold in the case of trajectories starting at periapsis. As for the differences between the two cases, we observe from the analysis of the $\eta=1$ plot that the benefits stemming from a faster convergence rate, previously important for decreasing $u_\text{sat,min}$ over moderate choices of $k_1$, are no longer significant when beginning at periapsis. This is a mere consequence of starting in a region where the nonlinear acceleration error is at a maximum. For this reason, from the perspective of the saturation threshold, there is no benefit for increasing $\eta$ when starting at periapsis, as $\eta=0$ provides the optimal $u_\text{sat,min}$ regions. Operational considerations may require a faster convergence speed and thus mandate the selection of a larger value for $\eta$, but this is considered to be outside of the scope of this work.
\begin{figure}[htpb]
	\centering
	\includegraphics[width=\linewidth]{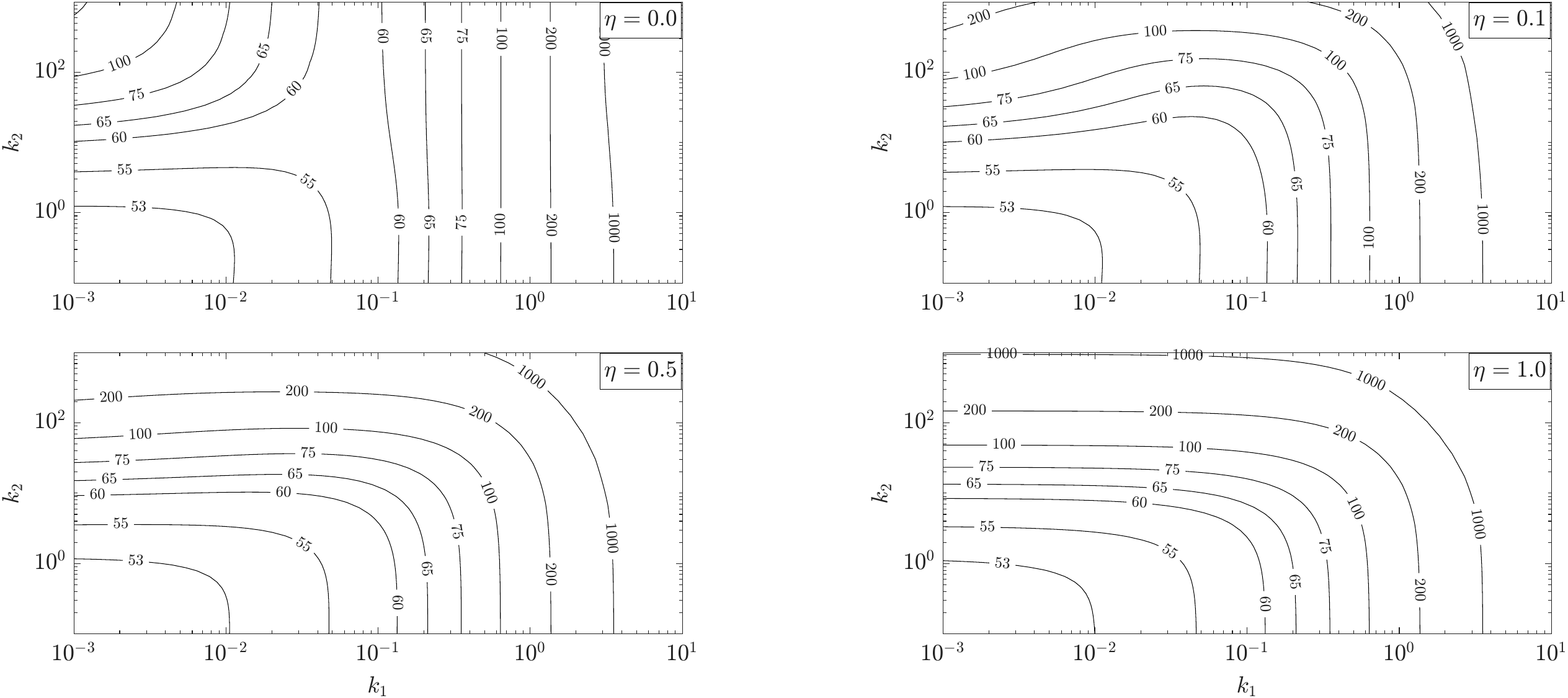}
	\caption{Level curves for the minimum saturation threshold $\frac{u_\text{sat,min}}{\norm{\mathbf{z}_0}}$, for the L2 Northern Halo QPO and $\norm{\mathbf{z}_0}\approx1\times 10^{-3}$, considering insertion at periapsis.}
	\label{fig:u_sat_L2_Halo_per}
\end{figure}

The conclusions drawn from the study of Figs.~\ref{fig:u_sat_L2_Halo_apo} and \ref{fig:u_sat_L2_Halo_per} may thus be summarized as follows: when possible, selecting insertion at apoapsis provides access to the more lenient saturation thresholds for the same control gains. Conversely, given a saturation threshold, insertion at apoapsis ensures that the gain selection may be made from a larger pool of options, with $\eta\to 0$ leading to the broadest pool. However, slightly increasing $\eta$ may prove beneficial to access desirable gain configurations by taking advantage of the faster minimum convergence rate. On the contrary, if insertion at periapsis is mandatory, optimal conditions are always reached for the maximum possible relief, i.e. taking $\eta\to 0$.

Given that $\eta$ is, in principle, a design choice, perhaps it is relevant to optimize its value for each gain pair $(k_1,k_2)$, namely with the intent of minimizing $u_\text{sat,min}$. Keeping the aforementioned considerations in mind, we perform this analysis only for insertion at apoapsis. For optimization purposes, we make recourse of the nonlinear least-squares routine provided by MATLAB's Optimization Toolbox \cite{matlabOptimizationToolbox}. The level curves resulting from this procedure are presented in Fig.~\ref{fig:u_sat_min_opt}. In addition, Figs.~\ref{fig:eta_opt} and \ref{fig:beta_min_opt} present the corresponding contour plots of the optimal $\eta$ and $\beta_\text{min}$ values, respectively.

\begin{figure}[htpb]
	\centering
	\includegraphics[width=\textwidth]{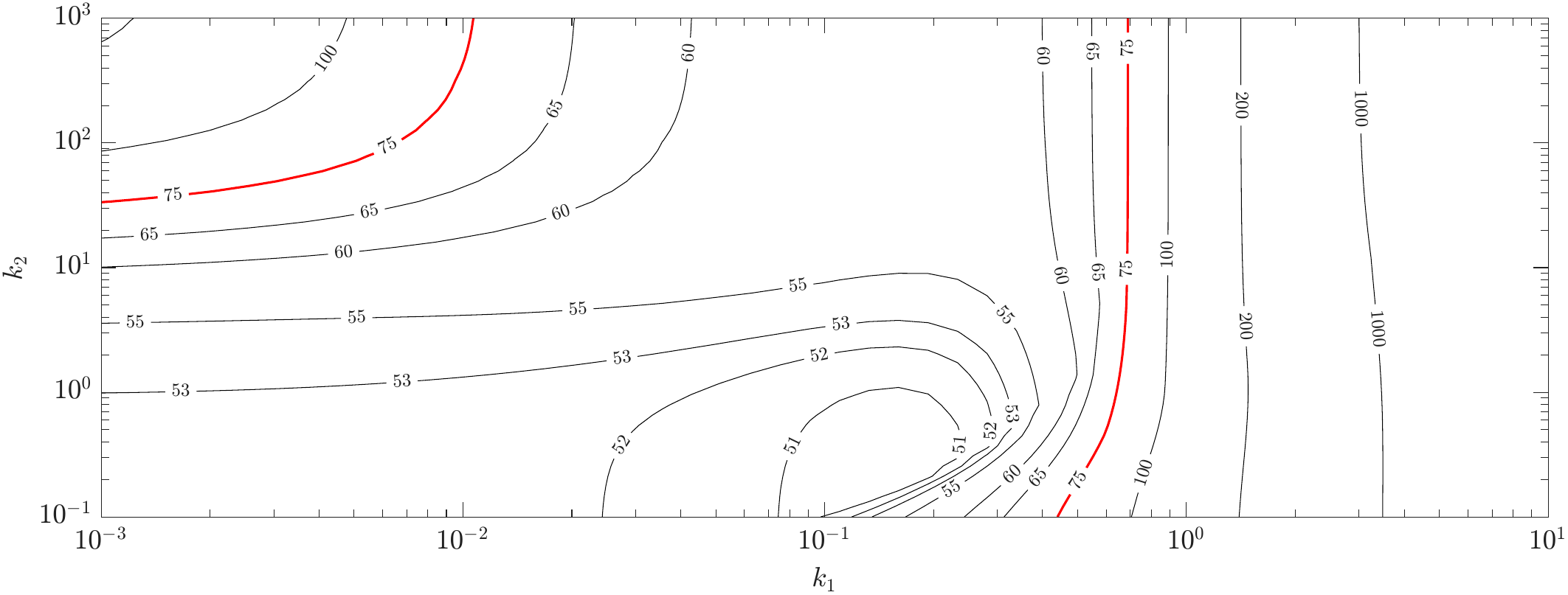}
	\caption{Level curves of the optimal minimum actuation saturation threshold $\frac{u_\text{sat,min}}{\norm{\mathbf{z}_0}}$, for the L2 Northern Halo QPO and $\norm{\mathbf{z}_0}\approx1\times 10^{-3}$, considering insertion at apoapsis. The highlighted level curve corresponds to a value of $u_\text{sat,min}\approx 0.1~\si{\newton}$ for a $500~\si{\kilogram}$ spacecraft.}
	\label{fig:u_sat_min_opt}
\end{figure}

\begin{figure}[htpb]
	\centering
	\begin{subfigure}{.48\textwidth}
		\centering
		\includegraphics[width=\linewidth]{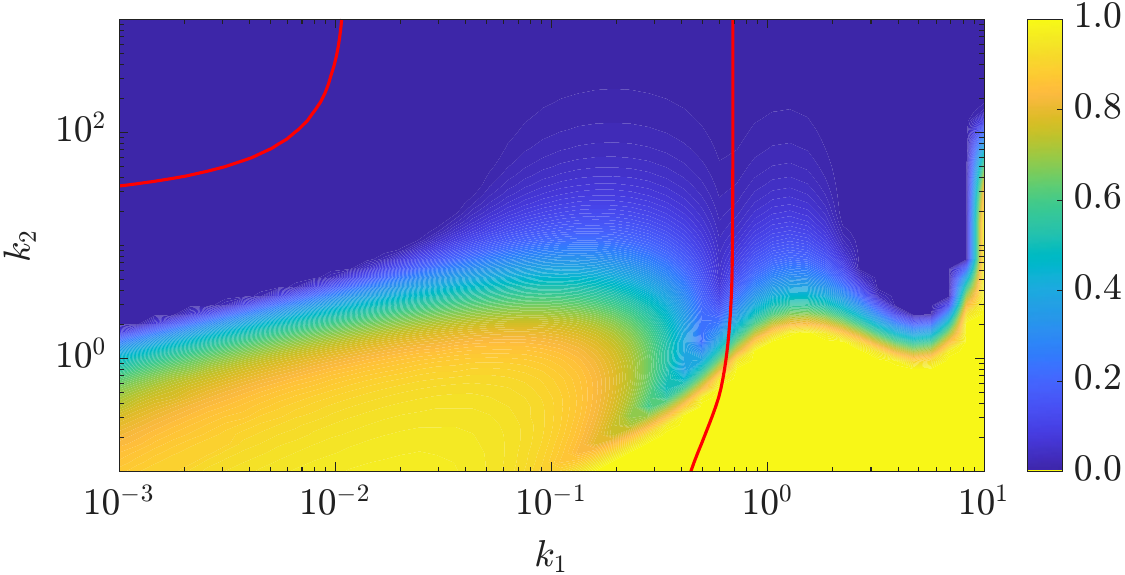}
		\caption{Optimal $\eta$ factor.}
		\label{fig:eta_opt}
	\end{subfigure}%
	\hfill
	\begin{subfigure}{.48\textwidth}
		\centering
		\includegraphics[width=\linewidth]{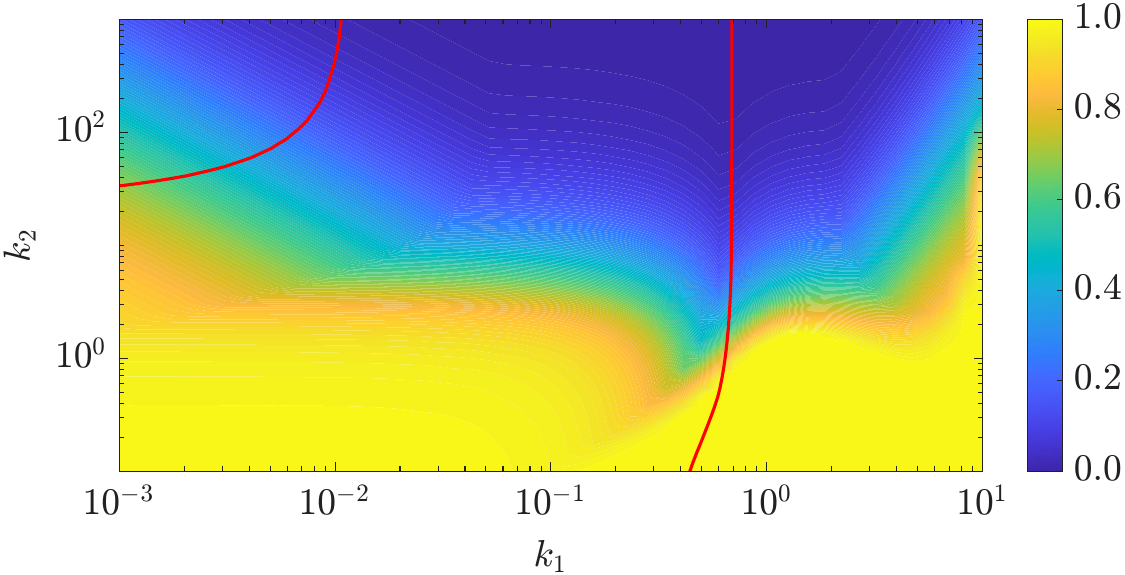}
		\caption{Maximum relief, $\beta_\text{min}$.}
		\label{fig:beta_min_opt}
	\end{subfigure}
	\caption{Optimal $\eta$ factor and maximum relief, $\beta_\text{min}$, leading to the results in Fig.~\ref{fig:u_sat_min_opt}, for the L2 Northern Halo QPO and $\norm{\mathbf{z}_0}\approx 1\times 10^{-3}$, considering insertion at apoapsis. The highlighted curve corresponds to the conditions yielding $u_\text{sat,min}\approx 0.1~\si{\newton}$, for a $500~\si{\kilogram}$ spacecraft.}
	\label{fig:beta_params_opt}
\end{figure}

For a given saturation threshold, Fig.~\ref{fig:u_sat_min_opt} may be used to retrieve optimal gain pairs that ensure $u_\text{sat}$ is respected without a sacrifice to the exponential stability guarantees, even if the controller saturates. In this work, we consider the saturation value to be $0.1~\si{\newton}$ -- a conservative figure in comparison with typical limits of hall-effect and ion based electric propulsion systems \cite{shastry2025ElectricPropulsion, Herman2016IonPropulsion} and the values typically considered in the literature for similar works \cite{McGuire2017LowThrustElectric, qi2019ContinuousThrust,du2023LowThrustER3BP}. Given the mass of the spacecraft detailed in Section~\ref{subsec:res_external_perts}, this threshold roughly corresponds to $u_\text{sat}\approx 75 \norm{\mathbf{z}_0}$, in non-dimensional units. In that case, the control gains may be selected from the region delimited by the level curve highlighted in Fig.~\ref{fig:u_sat_min_opt} and the corresponding actuation relief limit from Fig.~\ref{fig:beta_min_opt}. We highlight the degree of conservatism of the proposed implementation as, for example, the choice of $k_1=0.5=k_2$ -- previously considered in the analysis of the base control law, in Fig.~\ref{fig:MC_L2_Halo_apo} -- lies inside the specified region but the control command never reaches the saturation threshold. In fact, this is verified for most choices of $(k_1,k_2)$ inside the established region. Still, we argue that the formal projection of a minimum saturation threshold benefits from the proposed analysis since significant reliefs may be obtained for various choices of the gains, as illustrated in Fig.~\ref{fig:beta_min_opt}. This offers a decrease in terms of $u_\text{sat,min}$ when compared to a naive analysis that simply bounds the expected actuation command magnitude with no further considerations, which would lead to the level curves previously obtained for $\eta=1.0$.


While it has been established that the previous results for ${k_1=0.5=k_2}$ already formally respect the aforementioned saturation constraint, it is still relevant to search for illustrative scenarios where saturation is actually observed. For this purpose, we select another control gain pair from the limit condition identified by the $u_\text{sat,min}=75\norm{\mathbf{z}_0}$ level curve from Fig.~\ref{fig:u_sat_min_opt}.
In particular, the pair $(k_1,k_2)=(0.6962,300)$ is purposefully selected for its large value of $k_2$. This choice is made to excite the controller and force the occurrence of saturation for validation. In practice, however, a real mission would likely benefit from more modest gains, carefully selected by leveraging performance, robustness, and the saturation threshold. Once more, a Monte Carlo analysis totaling $100$ runs is carried out considering the operational constraints and errors from Table~\ref{tab:L2_Halo_op_params}, the external perturbations detailed in Section~\ref{subsec:res_external_perts} and, moreover, the established saturation threshold of $u_\text{sat}\approx 0.1~\si{\newton}$ -- equivalent to $2\times10^{-4}~\si{\meter \per \second}$ for the $500~\si{\kilogram}$ spacecraft considered. The corresponding system responses are plotted in Fig.~\ref{fig:MC_L2_Halo_apo_sat}, for the first $50$ days of simulation, with a detailed zoom on the control profile over the first hour. Note that, given the error superimposed on the control profile, the actual saturation threshold may differ slightly between each run.

\begin{figure}[htpb]
	\centering
	\includegraphics[width=\linewidth]{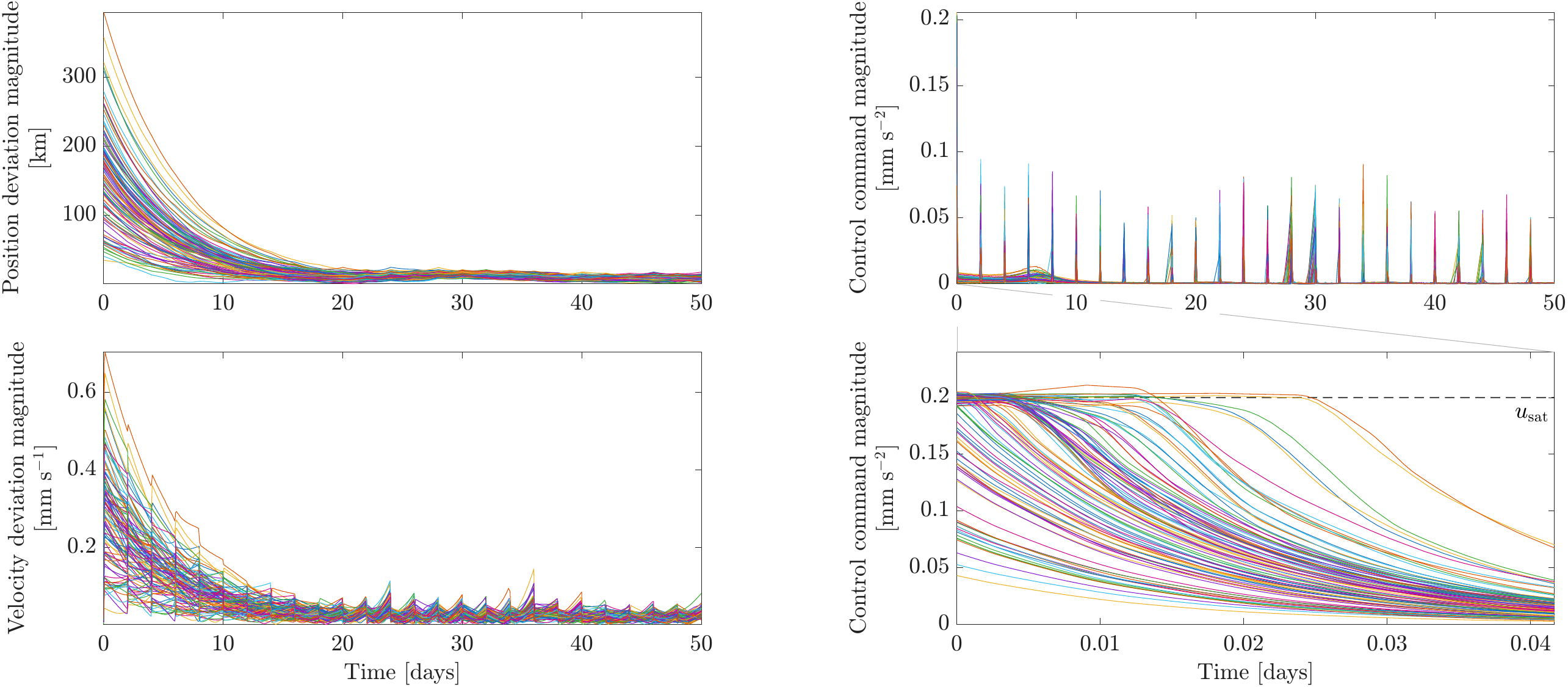}
	\caption{System responses of the Monte Carlo analysis of station-keeping about the L2 Northern Halo QPO under external perturbations, over the first $50$ days of simulation, starting at apoapsis. The control law with saturation is employed, taking $k_1=0.6962$ and $k_2=300$, considering a saturation threshold of $u_\text{sat}=2\times10^{-4}~\si{\meter \per \second}$.}
	\label{fig:MC_L2_Halo_apo_sat}
\end{figure}

The results presented in Fig.~\ref{fig:MC_L2_Halo_apo_sat} bring confirmation to the respect for the actuation saturation constraint, without a sacrifice to the convergence towards the nominal trajectory -- in fact, the large insertion error appears to be corrected equally as fast as the previous results from Fig.~\ref{fig:MC_L2_Halo_apo}. Given that the control gains have been chosen purposefully to force the occurrence of saturation, each new measurement leads to a pronounced increase in control command, manifested in the form of spikes on the respective plot. Still, the observation of the detailed zoom on the first hour of simulation shows that the largest spike, at insertion, respects the saturation constraint (up to the actuation error of $2\%$) and that some of the trials do actually saturate. Further analysis on the total deviation with time also shows that its magnitude is largest at insertion, confirming that the proposed gain choice yields a conservative envelope on the exponential stability and saturation guarantees, that thus hold for every segment of the orbit. Future work should assess if it is beneficial to recompute new control gains as new measurements become available, in hopes of optimizing the control response taking advantage of the initially decreasing deviation from the target trajectory.

\section{Conclusions}
\label{sec:conc}

This work proposes a nonlinear control law for station-keeping about libration point QPOs of the Earth-Moon system, stemming from the application of a backstepping technique. The base control law is shown to yield almost global guarantees for UES under a realistic, high-fidelity model of the dynamics. By pursuing a continuous strategy, as opposed to classic impulsive-based alternatives, the use of highly-efficient emerging technologies in terms of electric propulsion is made possible. A relief to the linear contributions of the base control law is pondered in order to formally adhere to a saturation constraint, improving applicability to these low-thrust systems. The improved formulation is shown to uphold exponential stability guarantees even in the event of saturation, under mild assumptions on the magnitude of the deviation.

Numerical simulations are carried out to validate the adequacy of the base control solution at tracking two different QPOs, considering different operational constraints and errors, as well as relevant external perturbations. Moreover, its performance is compared with other continuous and impulsive alternatives from literature, through a Monte Carlo analysis of various benchmarks in control expenditure, deviation envelope, and control idle time. This study attests to the quality of the proposed solution, which is able to achieve a performance that is on par with its continuous peers, even those actively striving for optimization. In this regard, the proposed strategy provides benefits in terms of computational efficiency.

The ramifications of the formal inclusion of saturation constraint are numerically studied for an L2 Northern Halo QPO, namely through the analysis of the theoretical lower bound on the saturation threshold with varying control gains. In particular, the maximum relief for exponential stability and insertion location are carefully analyzed. It is found that insertion at apoapsis provides the broadest and most adequate control gain pool, for each saturation threshold. Moreover, the best choice of maximum allowed relief is found to heavily depend on the specific gain pair selected, leveraging convergence speed with the time-dependence of the nonlinear acceleration error bound. To this end, the relief at each gain pair is numerically optimized so as to minimize the corresponding lower bound on the saturation threshold.  By considering a conservative saturation value from literature on electric space propulsion systems, an illustrative example considering a gain pair resulting from these optimization efforts is further assessed through simulations. Saturation is observed when answering a large insertion error, though without a sacrifice to the convergence capabilities of the proposed approach.

Future work should assess taking into account actuator performance in the established benchmarks to better reflect differences between high- and low-thrust options and expose key differences with classic impulsive station-keeping strategies. In regards to control under saturation, it is relevant to assess the definition of the control gains at each new measurement, to take even further advantage of the time-dependence of the proposed analysis. Moreover, given that thrust in space is essentially based on mass transfer, it would be relevant to assess a more complete propulsion model that considers mass depletion over time, which is expected to introduce relevant, non-negligible variations in thrust magnitude. Different approaches to relieve the original backstepping control law should also be evaluated, namely considering the attenuation of the nonlinear acceleration term or separate adjustments to the position and velocity deviation terms of the linear contribution. Finally, extensions to account for faults at the level of the sensors and actuators should be investigated, leveraging existing work on fault-tolerant control with applications to nonlinear backstepping (e.g., \cite{Bounemeur2023Fuzzy,Bounemeur2025FuzzySat}).

\appendix

\section{Sparsity of Colliding Trajectories}
\label{ap:collision_traj_proof}

In the proof of Proposition~\ref{prop:base_backstepping}, it is stated that the set of initial conditions leading to a collision in finite time with any celestial body -- a singularity in the error dynamics in Eq.~\eqref{eq:HFEM_variational} -- constitutes a set of measure zero in $\mathbb{R}^6$ . In this section, we provide a formal proof for this statement.

Under the nonlinear control law from Eq.~\eqref{eq:base_law}, the HFEM error dynamics in Eq.~\eqref{eq:HFEM_variational} become
\begin{equation}
	\label{eq:ap_lin_dyn}
	\begin{aligned}
		\dot{\mathbf{z}}_1(t) &= \mathbf{z}_2(t),\\
		\dot{\mathbf{z}}_2(t) & = - (\mathbf{I}_3 + \mathbf{K}_2 \mathbf{K}_1) \mathbf{z}_1(t) - (\mathbf{K}_1 + \mathbf{K}_2) \mathbf{z}_2(t),
	\end{aligned}
\end{equation}
over $\mathbf{z}(t)\in\mathcal{D}(t)$. For fixed gains, the system is linear time-invariant and its response, from the initial conditions $\mathbf{z}_0$, obeys
\begin{equation}
	\label{eq:ap_z_STM}
	\mathbf{z}(t) = \boldsymbol{\Phi}(t,0)\mathbf{z}_0 ,
\end{equation}
where $\boldsymbol{\Phi}(t,0)$ denotes the state-transition matrix (STM) of Eq.~\eqref{eq:ap_lin_dyn}, between time $0$ and $t$. As such, we may define the \textit{collision set}, i.e. the set of initial states that lead to a collision in finite time, as
\begin{equation*}
	\mathcal{C} = \left\{ \mathbf{z}_0\in \mathbb{R}^6~\vert ~\exists t\geq0, \exists j \in \{1,\dots,P\}, ~\text{such that}~ \left[\boldsymbol{\Phi}(t,0)\mathbf{z}_0 \right]^{(1:3)}= \mathbf{w}_j(t) \right\},
\end{equation*}
where $\mathbf{x}^{(1:3)}$ denotes the first three entries of $\mathbf{x}$, such that $\mathbf{z}_1(t)=[\boldsymbol{\Phi}(t,0)\mathbf{z}_0]^{(1:3)}$, following Eq.~\eqref{eq:ap_z_STM}.
We thus want to show that $\mathcal{C}$ has measure zero in $\mathbb{R}^6$.

We assume that each celestial body moves smoothly in configuration space and remains sufficiently far from each other. Moreover, the target trajectory is assumed to be collision-free, i.e. that it also remains sufficiently far from every celestial body. In that case, for each $j=1,\dots,P$, define the map
\begin{equation*}
	\boldsymbol{\psi}_j : \mathcal{S}_j\to\mathbb{R}^6, \quad \boldsymbol{\psi}_j(\mathbf{x},t) = \boldsymbol{\Phi}^{-1}(t,0)\mathbf{x}=\boldsymbol{\Phi}(0,t)\mathbf{x}, 
\end{equation*}
where
\begin{equation*}
	\mathcal{S}_j = \left\{ (\mathbf{x},t)\in\mathbb{R}^6\times [0,\infty): \mathbf{x}^{(1:3)}=\mathbf{w}_j(t)\right\}
\end{equation*}
denotes the set of all state-time pairs whose position coincides with the $j$-th singularity at time $t$.  In that case, we may take advantage of the properties of the STM to map all state-time pairs identifying a collision, encoded in each $\mathcal{S}_j$, towards their respective initial conditions at $t=0$. In other words, the collision set may be written compactly as
\begin{equation}
	\label{eq:ap_C}
	\mathcal{C} = \bigcup_{j=1}^{P} \boldsymbol{\psi}_j(\mathcal{S}_j).
\end{equation}Clearly, each $\mathcal{S}_j$ is a $C^1$ manifold of dimension $m=4$ in $\mathbb{R}^6\times[0,\infty)$, as it represents a smooth surface in velocity and time, under the position constraint. Moreover, the map $\boldsymbol{\psi}_j$ is also $C^1$, given the smooth and continuous nature of the linear dynamics in Eq.~\eqref{eq:ap_lin_dyn} and the invertibility of the STM, by definition \cite[Property 4.10]{rugh1996LinearSystems}.

Following \cite[Corollary 6.11]{lee2013introduction}, the image of a $C^1$ map $f: M\to N$ between manifolds $M$ and $N$ with respective dimensions $m$ and $n$, satisfying ${m<n}$, has measure zero in $N$. As such, each image $\boldsymbol{\psi}_j(\mathcal{S}_j)$ has measure zero in $\mathbb{R}^6$. Given that a countable union of measure zero sets also satisfies this property \cite[Proposition C.18]{lee2013introduction}, the collision set in Eq.~\eqref{eq:ap_C} has measure zero in $\mathbb{R}^6$, q.e.d.

\section{Addendum on the Critical Relief Factor}
\label{ap:sylvesters_crit}

When analyzing the stability guarantees provided by the relieved control law from Eq.~\eqref{eq:relieved_law}, through the Lyapunov candidate function in Eq.~\eqref{eq:base_lyap}, Sylvester's criterion is invoked, considering $\mathbf{K}_1=k_1\mathbf{I}_3$ and $\mathbf{K}_2=k_2\mathbf{I}_3$, with $k_1,k_2>0$. It is found that the relief factor, $\beta(t)$, must be positive and satisfy a set of conditions that may be written compactly as
\begin{equation*}
		\beta_1<\beta(t)<\beta_2,~\forall t\geq 0,
\end{equation*}
where
\begin{equation*}
	(\beta_1,\beta_2) = \begin{cases}
		\left(\dfrac{A-B}{C},  \dfrac{A+B}{C}\right), & k_1\neq 1,\\
		\left(\dfrac{1}{1+k_2}, \infty \right), & k_1=1,
	\end{cases}
\end{equation*}
with
\begin{equation*}
	A = k_1^4+2k_1 k_2+1, \quad B = 2\sqrt{(k_1 k_2+1)(k_1k_2+k_1^4)}, \quad \text{and} \quad C = (k_1^2-1)^2.
\end{equation*}
Clearly, $A,B,C>0$, for all $k_1,k_2>0$.

The discussion in Section~\ref{subsec:con_sat} relies on the assumption that $0<\beta_1<1$ and $\beta_2>1$, for all $k_1,k_2>0$. This is trivial to verify in the case of $k_1=1$. For $k_1\neq 1$, observe that
\begin{equation}
	\label{eq:ap_A_C}
	A-C = 2k_1k_2+2k_1^2>0 \implies A>C,
\end{equation}
such that $\beta_2>1$ is immediately satisfied, since $B>0$.
Moreover, given that $A>0$ is also satisfied, one has
\begin{equation*}
	A^2-B^2 = (k_1^4-1)^2>0 \implies A>B,
\end{equation*}
which leads to the conclusion that $\beta_1>0$. Finally, Eq.~\eqref{eq:ap_A_C} ensures ${A-C>0}$, such that
\begin{equation*}
		B^2-(A-C)^2 = 4k_1k_2(k_1^2-1)^2>0 \implies B> A-C \implies C>A-B,
\end{equation*}
leading to $\beta_1<1$. We therefore have $0<\beta_1<1$ and $\beta_2>0$ for all $k_1,k_2>0$, q.e.d.

\section{Acceleration Error Bound}
\label{ap:f_a_bound}

As discussed in Section~\ref{subsec:con_sat}, it is necessary to bound the deviation-induced acceleration error in the dynamics of Eq.~\eqref{eq:HFEM_variational} in order to derive analytical conditions for the respect of an actuation saturation constraint. In this section, we develop a nonlinear bounding expression for this contribution, with respect to the initial deviation magnitude, aiming for low conservatism. We do so by first exploiting key facts regarding $\ell_1$ and $\ell_2$ norms, namely by noticing
\begin{equation*}
	\begin{aligned}
		\norm{\mathbf{f}_a(t,\mathbf{z}_1(t))}&\leq \norm{\mathbf{f}_a(t,\mathbf{z}_1(t))}_1 =\sum_{j=1}^{P}\mu_j\norm{\frac{\mathbf{w}_j(t)-\mathbf{z}_1(t)}{\norm{\mathbf{w}_j(t)-\mathbf{z}_1(t)}^3} - \frac{\mathbf{w}_j(t)}{\norm{\mathbf{w}_j(t)}^3}}_1\\
		&\leq\sum_{j=1}^{P}\mu_j\left(\norm{\frac{\mathbf{z}_1(t)}{\norm{\mathbf{w}_j(t)-\mathbf{z}_1(t)}^3}}_1 + \norm{\frac{\mathbf{w}_j(t)}{\norm{\mathbf{w}_j(t)-\mathbf{z}_1(t)}^3} - \frac{\mathbf{w}_j(t)}{\norm{\mathbf{w}_j(t)}^3}}_1 \right),
	\end{aligned}
\end{equation*}
where $\norm{\mathbf{x}}_1$ denotes the $\ell_1$ norm of $\mathbf{x}$.
Doing so evidences that the magnitude of the acceleration error satisfies
\begin{equation}
	\label{eq:ap_f_a_bound_preliminary}
	\norm{\mathbf{f}_a(t,\mathbf{z}_1(t))}\leq\sum_{j=1}^{P}\mu_j \Big(\norm{\mathbf{z}_1(t)}_1 g_j(\mathbf{z}_1(t)) + \norm{\mathbf{w}_j(t)}_1 h_j(\mathbf{z}_1(t)) \Big),
\end{equation}
where
\begin{equation*}
	g_j(\mathbf{z}_1(t)) := \frac{1}{\norm{\mathbf{w}_j(t)-\mathbf{z}_1(t)}^3}, ~\mathbf{z}_1(t) \neq \mathbf{w}_j(t)
\end{equation*}
and
\begin{equation*}
	h_j(\mathbf{z}_1(t)) := \abs{g_j(\mathbf{z}_1(t)) - \frac{1}{\norm{\mathbf{w}_j(t)}^3}}, ~\mathbf{z}_1(t) \neq \mathbf{w}_j(t),~\mathbf{w}_j(t)\neq \mathbf{0}.
\end{equation*}
In that case, attention shifts to bounding each $g_j$ and $h_j$ functions. Naturally, given the nature of $h_j$, it is necessary for the target trajectory to be collision free, i.e. $\mathbf{w}_j(t)\neq\mathbf{0},~j=1,\dots,P$, over the time domain of analysis, which is assumed to be the case. Moreover, both functions require $\mathbf{z}_1(t)\neq\mathbf{w}_j(t),$ ${j=1,\dots,P},$ to be satisfied during that period. In this work, we enforce this condition by restricting the analysis to a region of proximity to the nominal trajectory, i.e. ${\norm{\mathbf{z}_1(t)}\leq \delta < \norm{\mathbf{w}_j(t)}}$, where $\delta$ is a design parameter to be detailed shortly. Through the triangle inequality, this means that
\begin{equation*}
	\norm{\mathbf{w}_j(t)}-\delta\leq\norm{\mathbf{w}_j(t)-\mathbf{z}_1(t)}\leq\norm{\mathbf{w}_j(t)}+\delta.
\end{equation*}
Under these conditions, it is clear that
\begin{equation}
	\label{eq:ap_g_bound}
	\frac{1}{(\norm{\mathbf{w}_j(t)}+\delta)^3}  \leq g_j(\mathbf{z}_1(t))\leq \frac{1}{(\norm{\mathbf{w}_j(t)}-\delta)^3}.
\end{equation}
Moreover, one may also write
\begin{equation*}
	\begin{aligned}
		h_j(\mathbf{z}_1(t))&= \max\left(g_j(\mathbf{z}_1(t))-\frac{1}{\norm{\mathbf{w}_j(t)}^3},\frac{1}{\norm{\mathbf{w}_j(t)}^3} - g_j(\mathbf{z}_1(t))\right)\\
		&\leq \max\left(\frac{1}{(\norm{\mathbf{w}_j(t)}-\delta)^3} -\frac{1}{\norm{\mathbf{w}_j(t)}^3},\frac{1}{\norm{\mathbf{w}_j(t)}^3} -\frac{1}{(\norm{\mathbf{w}_j(t)}+\delta)^3} \right).
	\end{aligned}
\end{equation*}
Since $f(x):=x^{-3}$ is a convex function over $x>0$, any secant between two points $a,b>0$ lives above the function's curve. In particular, for $a=x-y$ and $b=x+y$, with $x>y>0$, the middle point along the secant obeys
\begin{equation*}
	\frac{f(x+y)+f(x-y)}{2} \geq f(x) \implies \frac{1}{(x-y)^3}-\frac{1}{x^3}\geq \frac{1}{x^3}- \frac{1}{(x+y)^3}.
\end{equation*}
In that case, since $\norm{\mathbf{w}_j(t)}>\delta>0$, one has
\begin{equation*}
	\frac{1}{(\norm{\mathbf{w}_j(t)}-\delta)^3} -\frac{1}{\norm{\mathbf{w}_j(t)}^3} \geq \frac{1}{\norm{\mathbf{w}_j(t)}^3} -\frac{1}{(\norm{\mathbf{w}_j(t)}+\delta)^3}, 
\end{equation*}
such that
\begin{equation}
	\label{eq:ap_h_bound}
	h_j(\mathbf{z}_1(t)) \leq \frac{1}{(\norm{\mathbf{w}_j(t)}-\delta)^3}-\frac{1}{\norm{\mathbf{w}_j(t)}^3}.
\end{equation}

Evidently, the least conservative upper bounds in Eqs.~\eqref{eq:ap_g_bound} and \eqref{eq:ap_h_bound} are obtained for the lowest choice of $\delta$. If the conditions of Corollary~\ref{cor:ell_restriction_2} hold, recall that the exponential stability guarantees provide an upper bound on the magnitude of the total deviation vector, according to Eq.~\eqref{eq:UES_z_bound}, over ${t\in[0,t_s]}$, where $t_s>0$ is the final simulation time. Therefore, to ensure $\delta \geq\norm{ \mathbf{z}_1(t)}$ is satisfied over the period of analysis, the best choice for $\delta$ is, at each instant,
\begin{equation}
	\label{eq:ap_delta_optimal}
	\norm{\mathbf{z}_1(t)}\leq\norm{\mathbf{z}(t)}\leq\rho\norm{\mathbf{z}_0}e^{-\theta t} ~\leadsto~ \delta=\rho\norm{\mathbf{z}_0}e^{-\theta t},
\end{equation}
where $\rho=\sqrt{\frac{\lambda_\text{max}(\mathbf{X})}{\lambda_\text{min}(\mathbf{X})}}$, recalling the matrix $\mathbf{X}$ from the definition of the Lyapunov function in Eq.~\eqref{eq:base_lyap}, and $\theta$ is the minimum exponential convergence rate. In that case, for $\delta < \norm{\mathbf{w}_j(t)}$ to hold over $t\in[0, t_s]$, a sufficient condition to be met by the initial deviation is to belong to the open $\ell_2$ norm-ball
\begin{equation}
	\label{eq:ap_set_restriction}
	\mathcal{B} = \left\{\mathbf{x} \in \mathbb{R}^6: \rho \norm{\mathbf{x}}< r \right\}, \quad \text{with} \quad r=\min_{\substack{j=1,\dots,P\\ t\in[0, t_s]}} \norm{\mathbf{w}_j(t)}.
\end{equation}
In fact, since $ \mathbf{z}_0^T \mathbf{X} \mathbf{z}_0  \leq \lambda_{\max}(\mathbf{X}) \norm{\mathbf{z}_0}^2$ holds for $\mathbf{X}\succ 0$, one verifies that
\begin{equation*}
	\mathbf{z}_0\in\mathcal{B} \implies \mathbf{z}_0^T \mathbf{X} \mathbf{z}_0  < \lambda_{\max}(\mathbf{X})\frac{r^2}{\rho^2}= \lambda_{\min}(\mathbf{X})r^2 \implies \mathbf{z}_0 \in \mathcal{E},
\end{equation*}
such that the conditions of Corollary~\ref{cor:ell_restriction_2} are indeed satisfied. As such, any trajectory beginning at $\mathbf{z}_0\in\mathcal{B}$ is sure to evolve exponentially fast towards $\mathbf{z}=\mathbf{0}$, validating Eq.~\eqref{eq:ap_delta_optimal}. Moreover, while Eq.~\eqref{eq:ap_set_restriction} acts to restrict $\norm{\mathbf{z}_0}$, it may also be interpreted as a restriction to the control gain $k_1$ instead, through $\rho$, given an initial deviation $\norm{\mathbf{z}_0}$. In that case, $\norm{\mathbf{z}_0}<r$ must still be satisfied, i.e. $r$ must be sufficiently large. In practice, this means that the target trajectory must remain sufficiently far away from every celestial body, which usually represents only a mild assumption that is satisfied by most QPOs.

For $\mathbf{z}_0 \in \mathcal{B}$, substitution of Eqs.~\eqref{eq:ap_g_bound}, \eqref{eq:ap_h_bound}, and \eqref{eq:ap_delta_optimal} into Eq.~\eqref{eq:ap_f_a_bound_preliminary} leads to a final bound on the acceleration error in the form of
\begin{equation}
	\label{eq:ap_f_a_bound}
	\norm{\mathbf{f}_a(t,\mathbf{z}_1(t))} \leq \phi(\norm{\mathbf{z}_0}) := \max_{t\in[0,t_s]}\sum_{j=1}^P \mu_j \psi_j(\norm{\mathbf{z}_0},t), ~\forall t\in[0,t_s],
\end{equation}
with
\begin{equation*}
		\psi_j(\norm{\mathbf{z}_0},t) =
		\frac{\norm{\mathbf{w}_j(t)}_1}{(\norm{\mathbf{w}_j(t)}-\rho\norm{\mathbf{z}_0}e^{-\theta t})^3} - \frac{\norm{\mathbf{w}_j(t)}_1}{\norm{\mathbf{w}_j(t)}^3} + 
		\frac{\sqrt{3}\rho\norm{\mathbf{z}_0}e^{-\theta t}}{(\norm{\mathbf{w}_j(t)}-\rho\norm{\mathbf{z}_0}e^{-\theta t})^3},
\end{equation*}
where Cauchy-Schwarz inequality was used to write
\begin{equation*}
	\norm{\mathbf{z}_1(t)}_1\leq\sqrt{3}\norm{\mathbf{z}_1(t)}\leq \sqrt{3}\norm{\mathbf{z}(t)}\leq \sqrt{3} \rho \norm{\mathbf{z}_0}e^{-\theta t}.
\end{equation*}
This need not be done for $\norm{\mathbf{w}_j(t)}_1$ since it is precisely known for a given target trajectory, depending solely on the motion of the celestial bodies considered, which may be retrieved through their corresponding ephemeris. In that case, the value of Eq.~\eqref{eq:ap_f_a_bound} may be promptly determined through numerical means for a given initial deviation magnitude and choice of $k_1$, $k_2$, and $\beta_\text{min}$, which define $\rho$ and $\theta$.

We draw attention to a key fact about Eq.~\eqref{eq:ap_f_a_bound} and its relation with the minimum rate of convergence, $\theta$, specifically when tracking QPOs. Due to the nearly-repeating nature of these solutions, it is oftentimes the case that orbital insertion may be pondered at different locations along the trajectory, namely apoapsis or periapsis with respect to the closest celestial body. This has direct consequences at the level of each $\mathbf{w}_j(t)$, namely changing their respective initial value, which in turn alters the behavior of $\psi_j$ and the bound of the nonlinear acceleration error. In the limit case where $\theta=0$, it becomes clear that the insertion point is irrelevant, given how time-dependence is lost for every $\psi_j$. In that case, the closest approach to the nearest celestial body (or bodies) is expected to dominate $\phi(\norm{\mathbf{z}_0})$. However, if $\theta\neq 0$, it may be the case that a careful selection of the insertion point benefits $\phi(\norm{\mathbf{z}_0})$ by exploiting the time-dependence stemming from the exponential convergence guarantees. As an example, if the contribution due to the closest celestial body dominates Eq.~\eqref{eq:ap_f_a_bound}, insertion at apoapsis with respect to that body may prove ideal, as convergence towards the target orbit is expected to decrease the corresponding $\psi_j$ contribution before the next periapsis passing.

To illustrate the aforementioned considerations, consider a simple (non-realistic) example where $P=1$ and $\mathbf{w}_1(t)$ satisfies
\begin{equation}
	\label{eq:w_j_example}
	\norm{\mathbf{w}_1(t)}_1 = \norm{\mathbf{w}_1(t)}= 1+0.1\cos(t+\alpha).
\end{equation}
In a way, changing $\alpha$ is equivalent to selecting the insertion point along a periodic trajectory, with $\alpha=0$ identifying apoapsis and $\alpha=\pi$ periapsis. In that case, Fig.~\ref{fig:psi_example} provides the evolution of $\psi_1(t)$ over the equivalent of 5 revolutions, for $\alpha=0$ in Fig.~\ref{fig:psi_apo} and $\alpha=\pi$ in Fig.~\ref{fig:psi_per}, under different minimum convergence rates, $\theta$, considering $\mu_1=1$, $\rho=1$, and $\norm{\mathbf{z}_0}=0.1$.

\begin{figure}[htpb]
	\centering
	\begin{subfigure}{.48\textwidth}
		\centering
		\includegraphics[width=\linewidth]{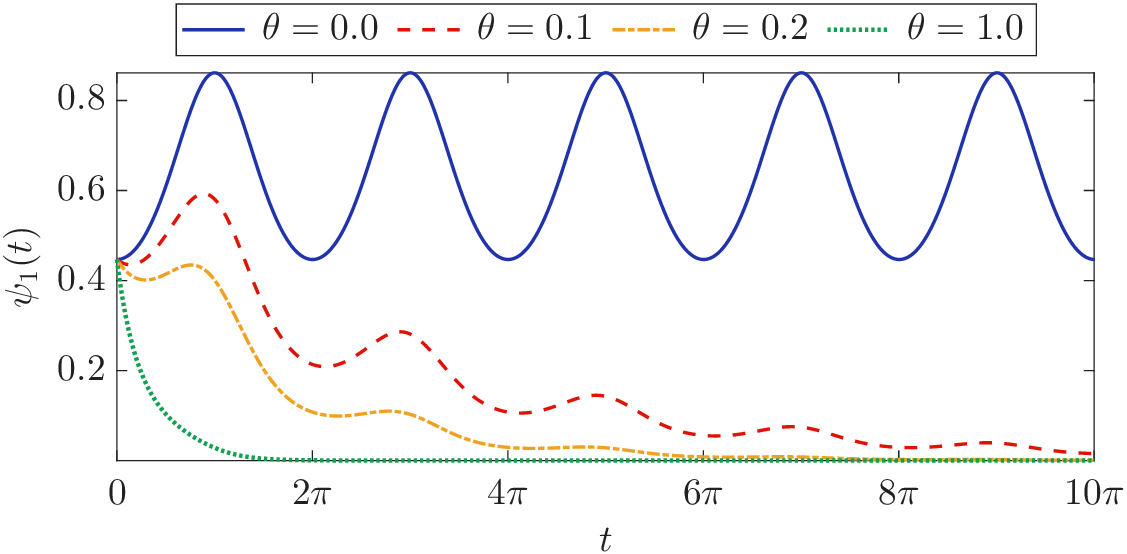}
		\caption{Insertion at apoapsis ($\alpha=0$).}
		\label{fig:psi_apo}
	\end{subfigure}%
	\hfill
	\begin{subfigure}{.48\textwidth}
		\centering
		\includegraphics[width=\linewidth]{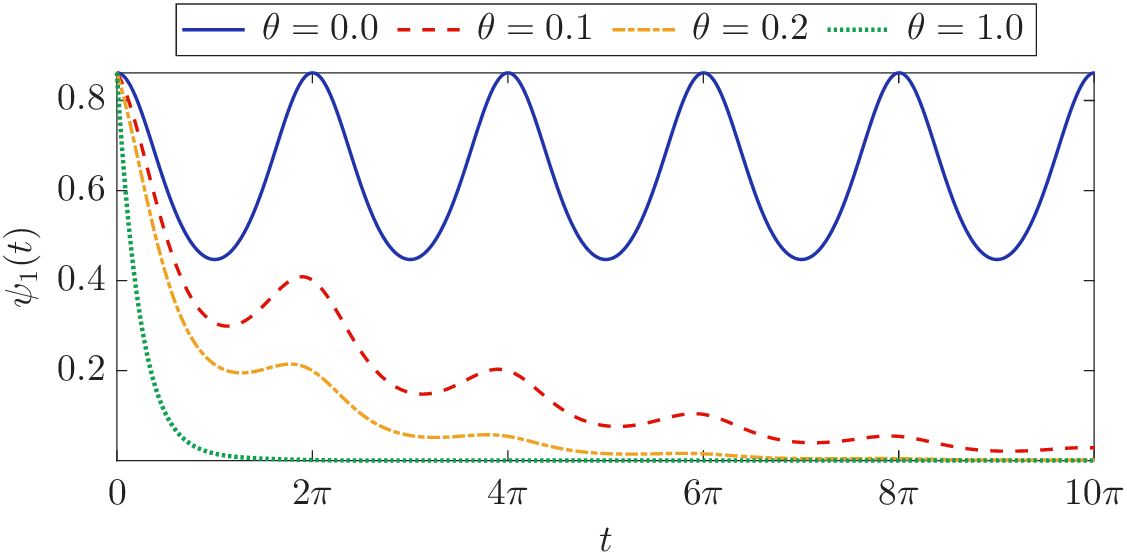}
		\caption{Insertion at periapsis ($\alpha=\pi$).}
		\label{fig:psi_per}
	\end{subfigure}
	\caption{Simple example of nonlinear acceleration error bound for a single contribution $\mathbf{w}_1(t)$, satisfying Eq.~\eqref{eq:w_j_example}, with $\mu_1=1=\rho$ and $\norm{\mathbf{z}_0}=0.1$, for various values of minimum rate of convergence, $\theta$.}
	\label{fig:psi_example}
\end{figure}

From the analysis of Fig.~\ref{fig:psi_example}, it becomes evident that the rate of convergence is only of significance if insertion is performed at a point where $\psi_1$ is reduced which, in this simple scenario, corresponds to apoapsis. In that case, it is possible to observe that even a relatively small exponential decay rate is sufficient to meaningfully decrease the maximum value of $\psi_1(t)$ achieved. For the test case presented, taking $\theta=0.2$ is already sufficient to minimize ${\phi=\max_t \mu_1\psi_1(t)}$, with larger values of $\theta$ yielding no further benefits. More realistic application scenarios are sure to introduce complexity and variety in $\psi_j(t)$. Still, it is oftentimes the case that a single celestial body (typically, the closest) dominates Eq.~\eqref{eq:ap_f_a_bound}, such that the insights provided by this simple example prove to be of broader relevance, as explored in Section~\ref{subsec:res_sat}.

\section*{Acknowledgments}
This work was supported by LARSyS FCT funding \linebreak(UID/50009/2025: DOI 10.54499/UID/50009/2025;\linebreak LA/P/0083/2020: DOI 10.54499/LA/P/0083/2020), and by\linebreak NEURASPACE Project, Contract No. 9, under Regulation (EU) 2021/241 of the European Parliament and of the Council of February 12, 2021 and the Portuguese Recovery and Resilience Program (PRR), in component 05 - Capitalization and Business Innovation, under Notice No. 425 01/C05-i01/2021 of the Regulation of Mobilizing Agendas/Alliances for re-industrialization.

\bibliographystyle{elsarticle-num} 
\bibliography{refs.bib}

\begin{thebibliography}{10}
\expandafter\ifx\csname url\endcsname\relax
  \def\url#1{\texttt{#1}}\fi
\expandafter\ifx\csname urlprefix\endcsname\relax\def\urlprefix{URL }\fi
\expandafter\ifx\csname href\endcsname\relax
  \def\href#1#2{#2} \def\path#1{#1}\fi

\bibitem{farquhar1970control}
R.~W. Farquhar, {The Control and Use of Libration-point Satellites}, Technical
  report, NASA (1970).

\bibitem{breakwell1974Stationkeeping}
J.~Breakwell, A.~Kamel, M.~Ratner, {Station-keeping for a translunar
  communication station}, Celestial Mechanics 10 (1974) 357--373.

\bibitem{wiesel1983ModalControl}
W.~{Wiesel}, W.~{Shelton}, {Modal control of an unstable periodic orbit},
  Journal of the Astronautical Sciences 31 (1983) 63--76.

\bibitem{simo1986StationkeepingInvManif}
C.~{Simó}, G.~{Goméz}, J.~{Llibre}, R.~{Martinez}, {Stationkeeping of a
  quasiperiodic halo orbit using invariant manifolds}, in: ESA Special
  Publication, Vol. 255 of ESA Special Publication, 1986, pp. 65--70.

\bibitem{simo1987OptimalControlHaloOrbits}
C.~Simó, G.~Gómez, J.~Llibre, R.~Martínez, J.~Rodríguez, On the optimal
  station keeping control of halo orbits, Acta Astronautica 15~(6) (1987)
  391--397.

\bibitem{dwivedi1975}
N.~P. Dwivedi, {Deterministic optimal maneuver strategy for multi-target
  missions}, Journal of Optimization Theory and Applications 17~(1-2) (1975)
  133--153.

\bibitem{howell1993StationkeepingPointMethod}
K.~C. Howell, H.~J. Pernicka, {Station-keeping method for libration point
  trajectories}, Journal of Guidance, Control, and Dynamics 16~(1) (1993)
  151--159.

\bibitem{folta2010StationKeeping}
D.~Folta, T.~Pavlak, K.~Howell, M.~Woodard, D.~Woodfork, {Stationkeeping of
  Lissajous Trajectories in the Earth-Moon System with Applications to
  ARTEMIS}, in: AAS/AIAA Space Flight Mechanics Meeting, San Diego, California,
  2010.

\bibitem{pavlak2012Strategy}
T.~Pavlak, K.~Howell, {Strategy for Optimal, Long-Term Stationkeeping of
  Libration Point Orbits in the Earth-Moon System}, in: AIAA/AAS Astrodynamics
  Specialist Conference, Minneapolis, Minnesota, 2012.

\bibitem{folta2013StationKeeping}
D.~Folta, T.~Pavlak, A.~Haapala, K.~Howell, M.~Woodard, {Earth–Moon libration
  point orbit stationkeeping - Theory modeling, and operations}, Acta
  Astronautica 94~(1) (2013) 421--433.

\bibitem{zhang2022StationKeepingHFEM}
R.~Zhang, Y.~Wang, Y.~Shi, C.~Zhang, H.~Zhang, {Performance analysis of
  impulsive station-keeping strategies for cis-lunar orbits with the ephemeris
  model}, Acta Astronautica 198 (2022) 152--160.

\bibitem{muralidharan2020ControlSF}
V.~Muralidharan, A.~Weiss, U.~Kalabic, {Control Strategy for Long-Term
  Station-Keeping on Near-Rectilinear Halo Orbits}, in: AAS/AIAA Space Flight
  Mechanics Meeting, Orlando, Florida, 2020.

\bibitem{ghorbani2013ContinuousAndImpulsive}
M.~Ghorbani, N.~Assadian, {Optimal station-keeping near Earth–Moon collinear
  libration points using continuous and impulsive maneuvers}, Advances in Space
  Research 52~(12) (2013) 2067--2079.

\bibitem{petersen2019JWST}
J.~Petersen, {L2 Station Keeping Maneuver Strategy for the James Webb Space
  Telescope}, in: AAS/AIAA Astrodynamics Specialist Conference, Portland,
  Maine, 2019.

\bibitem{lian2014SlidingMode}
Y.~Lian, G.~Gómez, J.~J. Masdemont, G.~Tang, {Station-keeping of real
  Earth–Moon libration point orbits using discrete-time sliding mode
  control}, Communications in Nonlinear Science and Numerical Simulation
  19~(10) (2014) 3792--3807.

\bibitem{Guzzetti2017stationkeeping}
D.~Guzzetti, E.~M. Zimovan, K.~C. Howell, D.~C. Davis, {Stationkeeping Analysis
  for Spacecraft in lunear Near Rectilinear Halo Orbits}, in: AAS/AIAA Space
  Flight Mechanics Meeting, San Antonio, Texas, 2017, pp. 3199--3218.

\bibitem{muralidharan2022StretchingDirections}
V.~Muralidharan, K.~C. Howell, {Leveraging stretching directions for
  stationkeeping in Earth-Moon halo orbits}, Advances in Space Research 69~(1)
  (2022) 620--646.

\bibitem{shastry2025ElectricPropulsion}
R.~{Shastry}, H.~{Kamhawi}, J.~D. {Frieman}, G.~C. {Soulas}, T.~G. {Gray},
  T.~R. {Verhey}, C.~D. {Kachele}, G.~J. {Williams}, N.~A. {Simmons}, R.~B.
  {Lobbia}, S.~M. {Arestie}, V.~H. {Chaplin}, J.~{Fisher}, G.~{Blackner},
  E.~{Forbes}, J.~{Hondagneu}, N.~A. {Branch}, J.~M. {Zubair}, H.~{Watts},
  {12-kW advanced electric propulsion system hall current thruster
  qualification and production status}, Journal of Electric Propulsion 4~(1)
  (2025) 61.

\bibitem{Herman2016IonPropulsion}
D.~A. Herman, W.~Santiago, H.~Kamhawi, J.~E. Polk, J.~S. Snyder, R.~R. Hofer,
  M.~J. Sekerak, {The Ion Propulsion System for the Asteroid Redirect Robotic
  Mission}, in: 52nd AIAA/SAE/ASEE Joint Propulsion Conference, Salt Lake City,
  Utah, 2016.

\bibitem{qi2019ContinuousThrust}
Y.~Qi, A.~{de Ruiter}, {Station-keeping strategy for real translunar libration
  point orbits using continuous thrust}, Aerospace Science and Technology 94
  (2019) 105376.

\bibitem{qi2022ContinuousThrust}
Y.~Qi, A.~de~Ruiter, {Trajectory correction for lunar flyby transfers to
  libration point orbits using continuous thrust}, Astrodynamics 6~(3) (2022)
  285--300.

\bibitem{nazari2014LQR_Backstepping}
M.~Nazari, W.~M. Anthony, E.~Butcher, {Continuous Thrust Stationkeeping in
  Earth-Moon L1 Halo Orbits Based on LQR control and Floquet Theory}, in:
  AIAA/AAS Astrodynamics Specialist Conference, San Diego, California, 2014.

\bibitem{jones1993H2}
B.~L. Jones, R.~H. Bishop, {$H_2$ optimal halo orbit guidance}, Journal of
  Guidance, Control, and Dynamics 16~(6) (1993) 1118--1124.

\bibitem{kulkrani2006}
J.~Kulkarni, M.~Campbell, G.~Dullerud, {Stabilization of Spacecraft Flight in
  Halo Orbits: An $H_\infty$ Approach}, IEEE Transactions on Control Systems
  Technology 14~(3) (2006) 572--578.

\bibitem{akiyama2018outputRegulation}
Y.~Akiyama, M.~Bando, S.~Hokamoto, {Station-keeping and formation flying based
  on nonlinear output regulation theory}, Acta Astronautica 153 (2018)
  289--296.

\bibitem{Nunes2025backstepping}
A.~Nunes, P.~Batista, S.~Brás, {Orbital Station-Keeping in the Earth-Moon
  System via Nonlinear Backstepping}, in: IEEE 19th International Conference on
  Control \& Automation, Tallinn, Estonia, 2025, pp. 75--80.

\bibitem{ulybyshev2015LinearProgramming}
Y.~Ulybyshev, {Long-Term Station Keeping of Space Station in Lunar Halo
  Orbits}, Journal of Guidance, Control, and Dynamics 38~(6) (2015) 1063--1070.

\bibitem{Elobaid2022MPC}
M.~Elobaid, M.~Mattioni, S.~Monaco, D.~Normand-Cyrot, {Station-Keeping of L2
  Halo Orbits Under Sampled-Data Model Predictive Control}, Journal of
  Guidance, Control, and Dynamics 45~(7) (2022) 1337--1346.

\bibitem{aerospace9120798}
S.~Cuevas~del Valle, H.~Urrutxua, P.~Solano-López, R.~Gutierrez-Ramon, A.~K.
  Sugihara, {Relative Dynamics and Modern Control Strategies for Rendezvous in
  Libration Point Orbits}, Aerospace 9~(12) (2022) 798.

\bibitem{shirobokov2017Survey}
M.~Shirobokov, S.~Trofimov, M.~Ovchinnikov, {Survey of Station-Keeping
  Techniques for Libration Point Orbits}, Journal of Guidance, Control, and
  Dynamics 40~(5) (2017) 1085--1105.

\bibitem{Nunes2026TrajectoryDesign}
A.~Nunes, S.~Brás, P.~Batista, J.~Xavier, {Designing trajectories in the
  Earth-Moon system: a Levenberg-Marquardt approach}, preprint submitted to
  Acta Astronautica (2026).
\newblock \href {https://doi.org/10.48550/arXiv.2510.18474}
  {\path{doi:10.48550/arXiv.2510.18474}}.

\bibitem{szebehely1967TheoryOfOrbits}
V.~Szebehely, {Theory of Orbits: the Restricted Problem of Three Bodies},
  Academic Press, 1967.

\bibitem{ross3BPbook}
W.~S. Koon, M.~W. Lo, J.~E. Marsden, S.~D. Ross, {Dynamical Systems, the
  Three-Body Problem, and Space Mission Design}, 2022.

\bibitem{poincare1892MethodesNouvelles}
H.~Poincaré, {Les méthodes nouvelles de la mécanique céleste}, Vol.~1,
  Gauthier-Villars et files, 1892.

\bibitem{gardner2006JWST}
J.~P. {Gardner}, J.~C. {Mather}, M.~{Clampin}, R.~{Doyon}, M.~A. {Greenhouse},
  H.~B. {Hammel}, J.~B. {Hutchings}, P.~{Jakobsen}, S.~J. {Lilly}, K.~S.
  {Long}, J.~I. {Lunine}, M.~J. {McCaughrean}, M.~{Mountain}, J.~{Nella}, G.~H.
  {Rieke}, M.~J. {Rieke}, H.-W. {Rix}, E.~P. {Smith}, G.~{Sonneborn},
  M.~{Stiavelli}, H.~S. {Stockman}, R.~A. {Windhorst}, G.~S. {Wright}, {The
  James Webb Space Telescope}, Space Science Reviews 123~(4) (2006) 485--606.

\bibitem{McComas2025IMAP}
D.~J. McComas, E.~R. Christian, N.~A. Schwadron, M.~Gkioulidou, F.~Allegrini,
  D.~N. Baker, M.~Bzowski, G.~Clark, C.~M.~S. Cohen, I.~Cohen, C.~Collura,
  M.~J. Cully, S.~Dalla, M.~I. Desai, A.~Driesman, D.~Eng, N.~J. Fox, H.~O.
  Funsten, S.~A. Fuselier, A.~Galli, J.~Giacalone, J.~Hahn, K.~P. Hegarty,
  T.~Horbury, M.~Horanyi, L.~M. Kistler, M.~A. Kubiak, S.~Kubota, S.~Livi,
  N.~Lugaz, C.~O. Lee, J.~Luhmann, W.~Matthaeus, D.~G. Mitchell, J.~G.
  Mitchell, E.~Moebius, S.~Pope, E.~Provornikova, J.~S. Rankin, D.~B.
  Reisenfeld, C.~Reno, J.~D. Richardson, C.~T. Russell, M.~M. Shaw-Lecerf,
  J.~Scherrer, R.~M. Skoug, M.~M. Shen, H.~E. Spence, Z.~Sternovsky,
  M.~Strumik, J.~R. Szalay, M.~Tapley, M.~Tokumaru, D.~L. Turner, S.~Weidner,
  J.~Westlake, P.~Wurz, G.~P. Zank, {Interstellar Mapping and Acceleration
  Probe: The NASA IMAP Mission}, Space Science Reviews 221~(8) (2025) 100.

\bibitem{li2021ChangE}
C.~Li, W.~Zuo, W.~Wen, X.~Zeng, X.~Gao, Y.~Liu, Q.~Fu, Z.~Zhang, Y.~Su, X.~Ren,
  F.~Wang, J.~Liu, W.~Yan, X.~Tan, D.~Liu, B.~Liu, H.~Zhang, Z.~Ouyang,
  {Overview of the Chang'e-4 Mission: Opening the Frontier of Scientific
  Exploration of the Lunar Far Side}, Space Science Reviews 217~(2) (2021) 35.

\bibitem{gomez2002SolarSystemFrequencies}
G.~Gómez, J.~J. Masdemont, J.~M. Mondelo, Solar system models with a selected
  set of frequencies, Astronomy and Astrophysics 390~(2) (2002) 733--749.

\bibitem{AlmanzaSoto2025PersistenceOfStructures}
J.-P. Almanza-Soto, K.~Howell, {Persistence of Restricted Three-Body Problem
  Normal Form Structures in Higher-Fidelity Models}, in: AAS/AIAA Astrodynamics
  Specialist Conference, Boston, Massachusetts, 2025.

\bibitem{park2025CharacterizationOfL2Analogs}
B.~Park, K.~C. Howell, {Characterization of Earth-Moon L2 halo analogs in an
  ephemeris model using the elliptic restricted three-body problem}, Advances
  in Space Research 75~(6) (2025) 5078--5109.

\bibitem{sanaga2024ChallengingRegion}
R.~R. Sanaga, K.~Howell, {Analyzing the Challenging Region in the Earth-Moon L2
  Halo Family via Hill Restricted Four-Body Problem Dynamics}, in: AIAA SCITECH
  Forum, Orlando, Florida, 2024.

\bibitem{khalil2001NLSystems}
H.~K. Khalil, Nonlinear Systems, 3rd Edition, Prentice-Hall, 2002.

\bibitem{tudatspace}
D.~{Dirkx}, M.~{Fayolle}, G.~{Garrett}, M.~{Avillez}, K.~{Cowan}, S.~{Cowan},
  J.~{Encarnacao}, C.~{Fortuny Lombrana}, J.~{Gaffarel}, J.~{Hener}, X.~{Hu},
  M.~{van Nistelrooij}, F.~{Oggionni}, M.~{Plumaris}, {The open-source
  astrodynamics Tudatpy software - overview for planetary mission design and
  science analysis}, in: European Planetary Science Congress, Granada, Spain,
  2022.

\bibitem{jpl_3bp_orbits}
{NASA Jet Propulsion Laboratory},
  \href{https://ssd.jpl.nasa.gov/tools/periodic_orbits.html}{{Three-Body
  Periodic Orbits}}, {Accessed: 2025-12-11} (2025).
\newline\urlprefix\url{https://ssd.jpl.nasa.gov/tools/periodic_orbits.html}

\bibitem{williams2017NRHO}
J.~Williams, D.~Lee, R.~Whitley, K.~Bokelmann, D.~Davis, C.~Berry, {Targeting
  Cislunar Near Rectilinear Halo Orbits for Human Space Exploration}, in:
  AAS/AAIA Space Flight Mechanics Meeting, San Antonio, Texas, 2017.

\bibitem{davis2017Stationkeeping}
D.~Davis, S.~Phillips, K.~Howell, S.~Vutukuri, B.~McCarthy, {Stationkeeping and
  Transfer Trajectory Design for Spacecraft in Cislunar Space}, in: AAS/AAIA
  Astrodynamics Specialist Conference, Stevenson, Washington, 2017.

\bibitem{Spreen2023BaselineNRHO}
E.~M. Zimovan-Spreen, S.~T. Scheuerle, B.~P. Mccarthy, D.~C. Davis, K.~C.
  Howell, {Baseline Orbit Generation for Near Rectilinear Halo Orbits}, in:
  AAS/AAIA Astrodynamics Specialist Conference, Big Sky, Montana, 2023.

\bibitem{matlabOptimizationToolbox}
{The MathWorks Inc.},
  \href{https://www.mathworks.com/help/optim/index.html}{{Optimization Toolbox
  version: 25.2 (R2025b)}}, {Accessed: 2026-21-01} (2025).
\newline\urlprefix\url{https://www.mathworks.com/help/optim/index.html}

\bibitem{McGuire2017LowThrustElectric}
M.~L. McGuire, L.~M. Burke, S.~L. McCarty, K.~J. Hack, R.~J. Whitley, D.~C.
  Davis, C.~Ocampo, {Low Thrust Cis-Lunar Transfers Using a 40 kW-Class Solar
  Electric Propulsion Spacecraft}, in: AAS/AIAA Astrodynamics Specialist
  Conference, Stevenson, Washington, 2017.

\bibitem{du2023LowThrustER3BP}
C.~Du, O.~Starinova, Y.~Liu, {Low-thrust transfer trajectory planning and
  tracking in the Earth-Moon elliptic restricted three-body problem}, Nonlinear
  Dynamics 111~(11) (2023) 10201--10216.

\bibitem{Bounemeur2023Fuzzy}
A.~Bounemeur, M.~Chemachema, {General fuzzy adaptive fault-tolerant control
  based on Nussbaum-type function with additive and multiplicative sensor and
  state-dependent actuator faults}, Fuzzy Sets and Systems 468 (2023) 108616.

\bibitem{Bounemeur2025FuzzySat}
A.~Bounemeur, M.~Chemachema, {Finite-time output-feedback fault tolerant
  adaptive fuzzy control framework for a class of MIMO saturated nonlinear
  systems}, International Journal of Systems Science 56~(4) (2025) 733--752.

\bibitem{rugh1996LinearSystems}
W.~J. Rugh, {Linear system theory}, 2nd Edition, Prentice Hall, 1996.

\bibitem{lee2013introduction}
J.~Lee, Introduction to Smooth Manifolds, 2nd Edition, Graduate Texts in
  Mathematics, Springer New York, 2013.

\end{thebibliography}
\end{document}